\documentclass[aps,prd,10pt,notitlepage,nofootinbib,superscriptaddress,showkeys,showpacs]{revtex4-1}

\usepackage{amsmath,amssymb,amsthm,latexsym,bbm}
\usepackage[english]{babel}
\usepackage{color}
\usepackage{xspace}
\usepackage{graphicx}
\usepackage{pifont,dsfont}
\usepackage{marvosym}
\usepackage{slashed}
\usepackage{enumitem}
\usepackage{multirow}
\usepackage{subcaption}
\usepackage{caption}
\usepackage{hyperref}

\theoremstyle{definition}
\newtheorem{lemma}{Lemma}

\newtheorem{theorem}{Theorem}
\newtheorem{remark}{Remark}

\newtheorem{proposition}{Proposition}

\makeatletter
\renewcommand\paragraph{\@startsection{paragraph}{4}{\z@}%
                                     {2ex\@plus .5ex \@minus -.2ex}%
                                     {-1ex \@plus .2ex}%
                                     {\normalfont\normalsize\bfseries}}
\makeatother

\newcommand{\ON}{$O(N)^3$-}
\begin{document}

%\title{{\Large \bf Comparison of the diagrammatics of two SYK-like tensor models}}
\title{{\Large \bf Diagrammatics of the quartic 
{\ON invariant Sachdev-Ye-Kitaev-like}
 tensor model}}

\author{{\bf Valentin Bonzom}}\email{bonzom@lipn.univ-paris13.fr}
\affiliation{LIPN, UMR CNRS 7030, Institut Galil\'ee, Universit\'e Paris 13,
99, avenue Jean-Baptiste Cl\'ement, 93430 Villetaneuse, France, EU}

\author{{\bf Victor Nador}}\email{victor.nador@ens-lyon.fr}
\affiliation{LaBRI, UMR CNRS 5800, Universit\'e de Bordeaux, 351 cours de la Lib\'eration, 33405 Talence, France, EU}

\author{{\bf Adrian Tanasa}}\email{ntanasa@u-bordeaux.fr}
\affiliation{LaBRI, Universit\'e de Bordeaux, 351 cours de la Lib\'eration, 33405 Talence, France, EU}
\affiliation{H. Hulubei Nat. Inst. Phys. Nucl. Engineering, P.O.B. MG-6, 077125 Magurele, Romania, EU}
\affiliation{I. U. F., 1 rue Descartes, 75005 Paris, France, EU}
\date{\today}

\begin{abstract}
Various tensor models have been recently shown to have the same properties as the celebrated Sachdev-Ye-Kitaev (SYK) model. In this paper we study in detail the diagrammatics of two such SYK-like tensor models: the multi-orientable (MO) model which has a $U(N)\times O(N)\times U(N)$ symmetry and 
%the CTKT model, 
a quartic $O(N)^3$-invariant model whose interaction has the tetrahedral pattern. We show that the Feynman graphs of the MO model can be seen as the Feynman graphs of the $O(N)^3$-invariant model which have an orientable jacket. Then we present a diagrammatic toolbox to analyze the $O(N)^3$-invariant graphs. 
{This toolbox}
allows for a simple strategy to identify all the graphs of a given order in the $1/N$ expansion. We apply it to the next-to-next-to-leading and next-to-next-to-next-to-leading orders which are the graphs of degree 1 and 3/2 respectively.
\end{abstract}

\maketitle

\section{Introduction}
\label{intro}

The fermionic quantum mechanical Sachdev-Ye-Kitaev (SYK) model (in the form introduced in \cite{kitaev}) has attracted, in the recent years, a huge amount of interest from the high energy physics community (see, for example, \cite{maldacena}, \cite{Gross}, \cite{PR}, \cite{Gurau_quenched} and references within).
A crucial diagrammatic property of the SYK model is that the model is dominated, in the large $N$ limit ($N$ being here the number of fermions), by a simple class of graphs called melonic graphs. A combinatorial proof of the melonic dominance of the SYK model has been recently given in \cite{Nador1}.

In \cite{witten}, Witten related the SYK model to the so-called colored tensor model, model originally introduced and extensively studied in the works of Gur\u au and collaborators (see the book \cite{book} and references within). This SYK-like tensor model is known today as the Gur\u au-Witten model.

In \cite{KT}, Klebanov and Tarnopolsky related the SYK model to another tensor model with an \ON symmetry\footnote{Let us emphasize that the symmetry of  tensor models studied here is  $O(N)^3$ and not  $O(N)$. Tensor models invariant under the action of the  $O(N)$ group have been recently studied in \cite{Kton}, \cite{carrozza2} and \cite{Carrozza3}.} ($N$ being here the size of the tensor in each entry). 
This model, whose purely combinatorial part was originally introduced
{by Carrozza and Tanasa} 
 in \cite{CTKT}, will be referred in this article 
{(as it is already referred in a part of the literature)}  
 as 
{the Carrozza-Tanasa-Klebanov-Tarnopolsky}  
 (CTKT), or the \ON-invariant model (although it is only one, with quartic interactions, of all possible \ON-invariant models).

Also in \cite{KT}, the SYK model was related to another tensor model having a $U(N)\times O(N)\times U(N)$ symmetry, model whose combinatorics is known as the multi-orientable (MO) tensor model, and studied in a series of several papers \cite{MO_original}, \cite{MO_expansion}, \cite{Fusy1}, \cite{donald} (see also the review paper \cite{MO_review}).

It is worth emphasizing here that all these tensor models are known to have a well-defined large $N$ limit (dominated by the same melonic graphs as the ones which dominate the large $N$ limit of the SYK model). Tensor models thus enlarge the club of models known to indeed have large $N$ limits, which already consisted in vector and matrix models -- see  the TASI lectures \cite{KT-review} for a recent review.

For the sake of completeness, let us also mention that the SYK model is know to be invariant under the action of the orthogonal group. However, SYK-like quantum mechanics with $Sp(N)$ symmetry has been recently investigated in \cite{Carrozza}.

\medskip

Thus, based on earlier works in the tensor model literature, reference \cite{KT} proposed two SYK-like models whose combinatorics are the MO model and the \ON-invariant model. We will analyze both from the purely combinatorial point of view of the Feynman graphs.

It is known that the MO graphs form a subset of \ON-invariant graphs \cite{CTKT}. In Theorem \ref{thm:Bijection} we identify this subset as the graphs with an orientable jacket. Jackets are key objects which are canonically obtained from the Feynman graphs (see \cite{MO_expansion} and \cite{CTKT}). They are ribbon graphs so one can make use of the notion of genus of these ribbon graphs\footnote{Let us also mention that, from a mathematical point of view, jackets can be seen as Heegaard surfaces canonically associated to tensor graphs \cite{Ryan}{.}}. It is actually the sum of the genera of the jackets which defines the degree. The latter is a positive half-integer which actually controls the $1/N$ expansion of tensor models.

Then we present a set of diagrammatic techniques to analyze the \ON-invariant graphs of the CTKT model. They converge to a strategy to identify all graphs of fixed degree. The strategy is to distinguish first the 2-particle-reducible (2PR) graphs from the 2-particle-irreducible (2PI) graphs. We recall that 2PR graphs are graphs with a 2-edge-cut and 2PI graphs are those without any 2-edge-cuts. The 2PR graphs of degree $\omega$ are easily obtained from those of smaller degrees. Then the strategy to find the 2PI graphs of fixed degree is to distinguish those with and without dipoles. Remarkably, all steps of the strategy except one only require to know the graphs of smaller degrees. The exception consists in finding the graphs of fixed degree which are 2PI and dipole-free. This step requires an independent analysis.

Graphs of degree 0 and of degree 1/2 are leading and next-to-leading order and consist in melonic graphs and tadpole graphs respectively. We apply our strategy to find all the graphs of degree 1, which are the next-to-next-to-leading order graphs. Theorem \ref{thm:Degree1} gives the unique 2PI, dipole-free graph of degree 1, then Theorem \ref{thm:Degree1AllGraphs} gives all the graphs of degree 1. We go as far as the next order in the $1/N$ expansion, which are the graphs of degree 3/2. Theorem \ref{thm:2PIDipoleFreeDegree3/2} shows that there is a unique, explicit graph of degree 3/2 which is 2PI and dipole-free. The other steps of the procedure to find all the graphs of degree 3/2 
{follow the exact same line of reasoning and are thus}
 left to the reader.

\medskip

The paper is organized as follows. As jackets and their orientability are
{a}
 key 
{ ingredient}  
 in our work, we present a quick review of the orientability of discrete surfaces in Section \ref{sec:Orientability}, both from the point of view of ribbon graphs and 
{of} 
 3-colored graphs, which are the two relevant representations of discrete surfaces here. Then we present the two SYK-like models in Section \ref{sec:Models}. We prove that MO graphs are \ON-invariant graphs with an orientable jacket in Section \ref{sec:Relation}.
Then in Section \ref{sec:tehnici} we present our diagrammatic toolbox for the analysis of \ON-invariant graphs, ending in {sub}section \ref{sec:Strategy} with our strategy to find all graphs of any given fixed degree. In Section \ref{sec:Degree0} we recall the graphs of degree 0 and 1/2. In Section \ref{sec:degre1}, we apply our strategy in detail and find all graphs of degree 1. 
Our last result is 
{exhibited} 
in Section \ref{sec:Degree3/2} which gives the unique graph of degree 3/2 which is 2PI and dipole-free.

%we define the two SYK-like tensor models studied in this paper, namely the MO and the $O(N)^3$-invariant model.  In the third section we state and prove our result characterizing, with respect to the genera of the graph's jackets, non-MO $O(N)^3$-invariant graphs.  In the rest of the paper we study the $O(N)^3$-invariant model.
%In Section $4$ we give some general diagrammatic techniques (removal of tadpoles, dipole insertions {\it etc.}) that we will use in the sequel.  The following section recalls the degree $0$ and $1/2$ 
%Sections $6$ and resp. $7$ are dedicated to the detailed study of degree $1$ and resp. $3/2$ $O(N)^3$-invariant graphs. We end the paper with a conclusion section.

%%%%%%%%%%%%% 
\section{Short review of orientability for discrete surfaces} \label{sec:Orientability}
%%%%%%%%%%%%%

A ribbon graph is a graph whose edges and vertices are thickened as ribbons. In addition to vertices and edges, it has faces which are the connected components of the ribbon complement. A useful representation is in terms of 2-stranded graphs, where the ribbon is drawn as a couple of strands which delimits its edges. The faces are then identified as the closed strands.

A ribbon graph thus encodes a discrete surface, whose genus is given by Euler's formula. The standard representation of ribbon graphs is to draw all ribbon vertices as flat road crossings,
\begin{equation}
\begin{array}{c} \includegraphics[scale=.5]{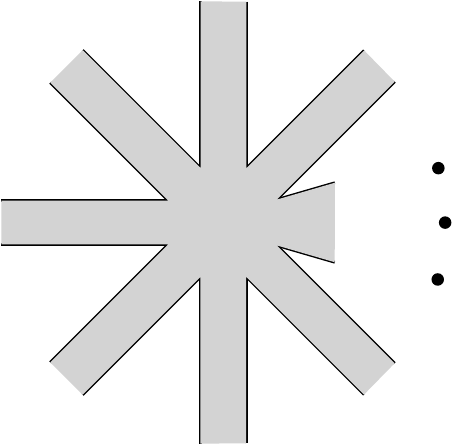} \end{array}
\end{equation}
Along an edge, the ribbon can be either flat too or it can be twisted,
\begin{equation}
\begin{array}{c} \includegraphics[scale=.5]{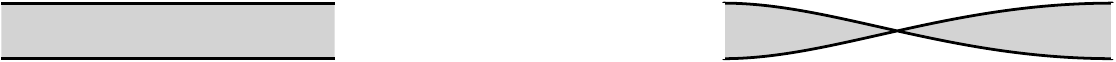} \end{array}
\end{equation}
The number of twists is modulo 2. A non-loop edge can be twisted and untwisted while taking the mirror image of one of its end vertices and twisting/untwisting the other incident non-loop edges. Two ribbon graphs are equivalent if they differ by a sequence of such operations. We refer to \cite{MoharThomassen} for all topological aspects of ribbon graphs, in particular the following theorem.

\begin{theorem} \label{thm:OrientableRibbon}
A ribbon graph without twist encodes an orientable surface.
\end{theorem}

Another representation of discrete surfaces which we will encounter is as edge-colored, 3-valent graphs, whose edges carry the colors $\{0,1,2\}$ and every vertex has exactly every color incident once. A discrete surface is reconstructed by using as faces the bicolored cycles of the graph, i.e. the cycles with colors $\{0,1\}$, $\{0, 2\}$, $\{1,2\}$. Notice that any ribbon graph can be transformed into an edge-colored graph by barycentric subdivisions, both encoding the same surface. In the case of edge-colored graph-encoded surface, the theorem about orientability is the following \cite{BipartitenessItalian}.

\begin{theorem} \label{thm:OrientableGraph}
An edge-colored graph is bipartite if and only if the corresponding surface is orientable.
\end{theorem}

An illustration is given in Figure \ref{fig:Orientability}. %In the bipartite case, it is easy to turn a colored graph into a ribbon graph. First embed the graph such that the clockwise order around a white vertex is $(0 1 2)$ and $(0 2 1)$ around a black vertex. Then simply thicken this graph without twists. Notice that the graphs of Figure \ref{fig:Orientability} are the only 3-colored graphs with four vertices, up to color permutation for the one on the left.
\begin{figure}
\includegraphics[scale=.5]{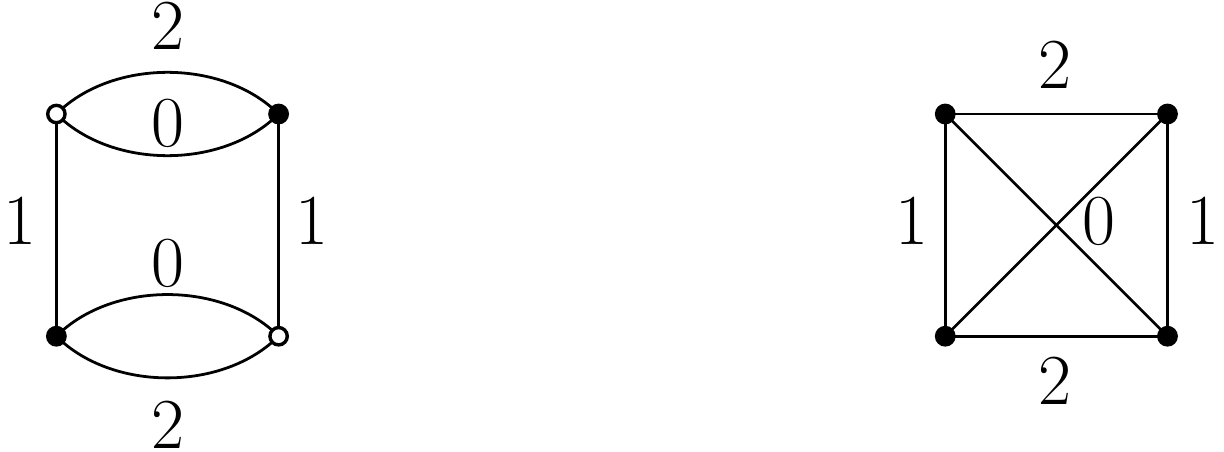}
\caption{\label{fig:Orientability} On the left is a bipartite 3-colored graph representing the sphere and on the right a non-bipartite graph representing the projective plane.}
\end{figure}

\begin{remark} \label{thm:CanonicalEmbedding}
\emph{Canonical embedding.} In the bipartite case, there is a natural way to construct 
{a}  
ribbon graph of the same genus, for which the bicolored cycles become the faces. 
{This graph} 
 is obtained by using a cyclic ordering of the colors, say $(012)$, at white vertices and the other ordering, $(021)$ at black vertices. In practice, one can draw the edges around a white vertex with the colors $(012)$ counter-clockwise, and the other way around for black vertices. Thickening this embedded graph gives a ribbon graph without twist.

Any non-bipartite colored graph, of half-integer genus $g$, can still be embedded without crossing on a surface of half-integer genus $g$ in such a way that the bicolored cycles become the boundary of faces homeomorphic to discs and such that the union of those faces is the surface itself. In practice, one can always start from representing the bicolored cycles with colors $\{1,2\}$ as polygons, with faces on the interior. Then the color 0 is added on the outside of the polygons of colors $\{1,2\}$.
\end{remark}

%%%%%%%%%%%%%%%%
\section{Two SYK-like tensor models} \label{sec:Models}
%%%%%%%%%%%%%
In this section, we present the two SYK-like tensor models studied thereafter, namely the $O(N){^3}$-invariant model  and the multi-orientable model (MO model). For each model, we give a brief presentation of its properties and discuss the structure of its Feynman graphs. The two models 
{are}
 studied 
for rank three tensor fields.
%but admit higher dimension generalization.

%%%%%%%%%%%%%%
\subsection{The \ON-invariant model, or { the} CTKT model} \label{sec:ON}
%%%%%%%%%%%%%%

As already mentioned in the introduction, the $O(N)^3$-invariant model was initially introduced in \cite{CTKT} in $2015$; in $2016${,}  in \cite{KT}, this model was related to the SYK model. This model consist of a single real fermionic tensor $\psi_{ijk}$ of size $N$ with $O(N)$ symmetry on each of its indices:
\begin{equation}
\psi_{abc} \rightarrow \psi'_{a'b'c'} = {O_1}_{a'}^{a}{O_2}_{b'}^{b}{O_3}_{c'}^{c} \psi_{abc}, \hspace{20pt} O_i \in O(N){.}
\label{symmetry_CTKT}
\end{equation}
With this symmetry, we can build two different invariants of order 4 in $\psi_{ijk}$ depending on how indices of the tensor are contracted. We can either have a {\bf tetrahedral} interaction,
\begin{equation}
I_1 = \psi_{abc}\psi_{ade}\psi_{fbe}\psi_{fdc}
\label{tetra_int}
\end{equation}
or a ``pillow'', also known as {\bf melonic} interaction,
\begin{equation}
I_2 = \psi_{abc}\psi_{dbc}\psi_{aef}\psi_{def}.
\label{pillow_int}
\end{equation}
These interaction terms can be represented diagrammatically by edge-colored graphs with four vertices of degree $3$ given in Figure \ref{CTKT_bubbles}. This diagrammatic correspondence is such that each vertex represents a tensor and each incident edge represents an index. Each edge carries a color corresponding to its position on the tensor. An edge of color $i$ connecting two vertices denotes a summation on the $i$-th index between two tensors. This makes sure that the interactions are \ON-invariant. More generally, $O(N)^d$- and $U(N)^d$-invariant polynomials are called bubbles in the literature.

Notice that $I_1$ is fully symmetric under permutations of the colors, while there are in fact three different versions of $I_2$ (where the two dipoles can be connected by edges of color 1, or 2, or 3).

\begin{figure}[ht]
\centering
\includegraphics[scale=0.65]{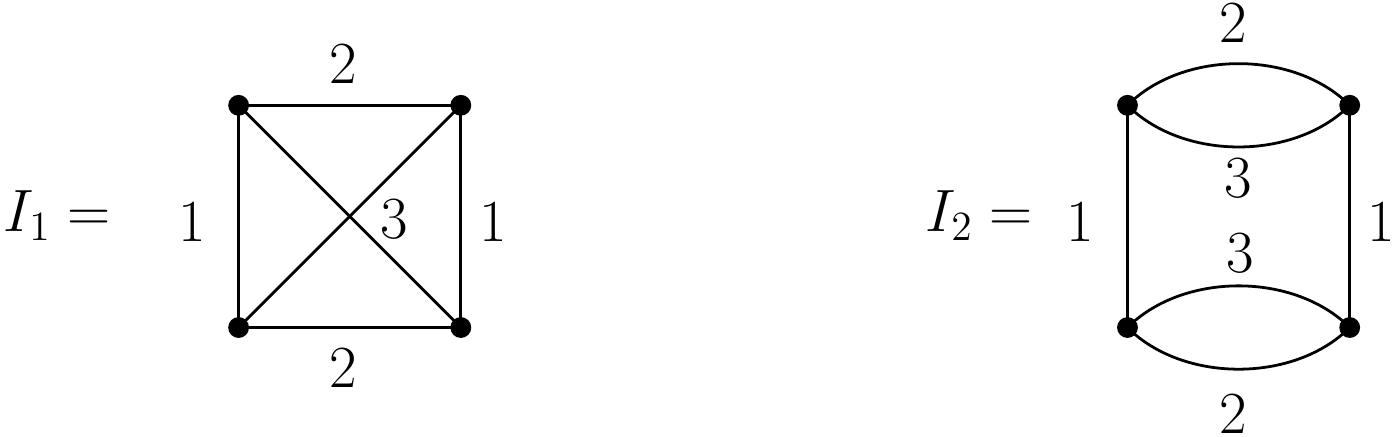}
\caption{On the left, the tetrahedral bubble. On the right, a melonic bubble with 4 vertices.}
\label{CTKT_bubbles}
\end{figure}

In this paper{,}  as in \cite{KT}, we only study the tetrahedral interaction. The SYK-like 
{$(0+1)-$dimensional}
action writes
\begin{equation}
\label{ctkt}
S_{CTKT}=\int dt \left( \frac{\imath}{2} \psi_{abc} \partial_t \psi_{abc}+ \frac{\lambda}{4}\psi_{abc}\psi_{ade}\psi_{fbe}\psi_{fdc}\right).
\end{equation}

Therefore, a Feynman graph of the model is represented as bubbles connected by propagators, which are represented as dashed edges, also referred to as $0$ colored edge. Therefore, all the vertices of the Feynman graph are of valency $4$ and have exactly one half-edge of each color $i\in\{0,1,2,3\}$. 

\begin{figure}[ht]
\centering
\includegraphics[scale=0.55]{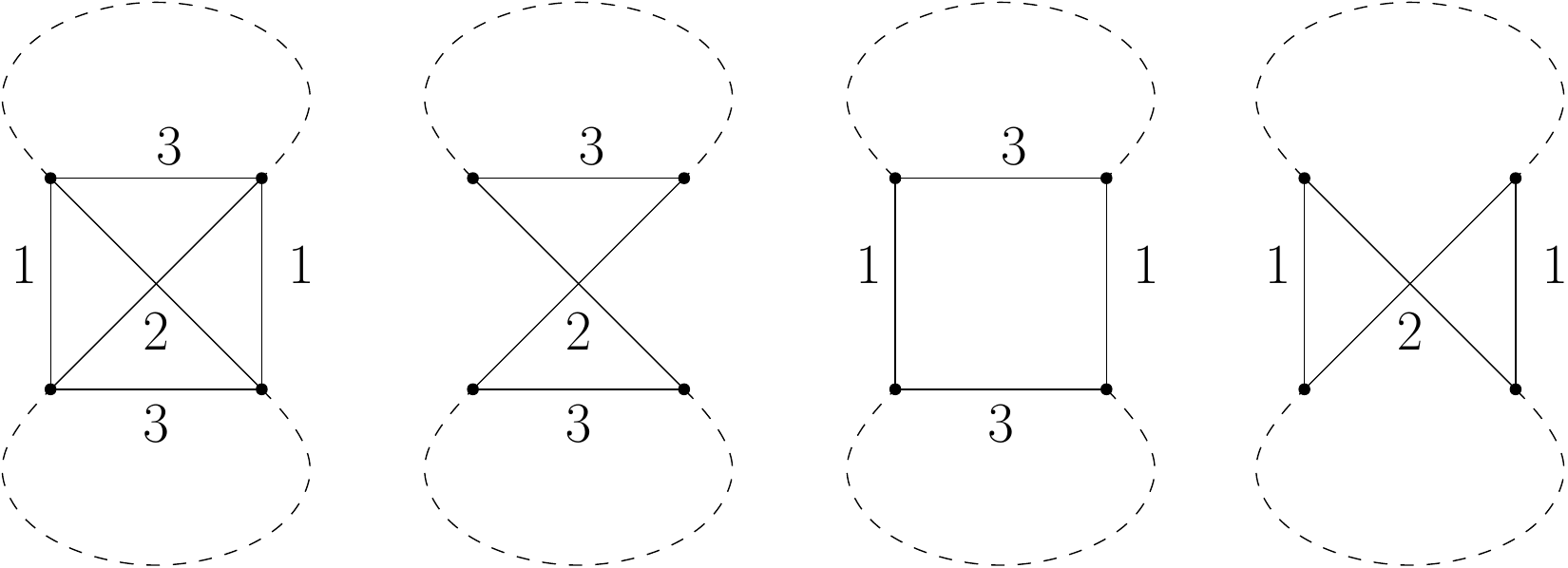}
\caption{An exemple of a Feynman graph of the CTKT model and its three jackets}
\label{CTKT_jacket}
\end{figure}

The {\bf jacket} $J_i$ of a graph $G$, for $i\in\{1,2,3\}$, is the graph obtained by deleting all edges of color $i$. If $j, k$ denotes the complementary colors, $\{i, j, k\}=\{1,2,3\}$, then $J_i$ is a 3-colored graph whose vertices are those of $G$ and have degree 3, and whose edges have colors $0, j, k$. An example of a Feynman graph and its three jackets is given in Figure \ref{CTKT_jacket}. Due to the structure of the tetrahedral bubble, all the jackets of a connected graph are connected. Therefore, a jacket represents a connected surface, whose genus is given by the Euler characteristic formula,
% and can be given a genus through Euler's characteristic:
\begin{equation}
\chi_i = 2-2g_i = V - E(J_i) + F(J_i),
\end{equation}
where $g_i$ can be a half-integer. Here $V$ is the number of vertices of $G$, $E(J_i)$ the number of edges of $J_i$ and $F(J_i)$ the number of faces of $J_i$. For example, the jackets of Fig. \ref{CTKT_jacket} have $V=4$, $E=6$.

Since $J_i$ has vertices of degree 3, $E(J_i) = 3V/2$. Moreover, for graphs encoding surfaces, the faces correspond to the bicolored cycles, with color pairs $\{0,j\}$, $\{0, k\}$ and $\{j,k\}$. Those bicolored cycles can be read either on $J_i$ or directly on $G$. Notice that the bicolored cycles with colors $\{j,k\}$ lie within the bubbles and there is exactly one for each bubble. Denoting $F_i, F_j, F_k$ the number of bicolored cycles with colors $\{0,i\}$, $\{0,j\}$, $\{0,k\}$, we have
\begin{equation}
F(J_i) = F_{j} + F_{k} + V/4.
\end{equation}
We then get
\begin{equation} \label{Euler_charac}
2-2g_i = F_j + F_k - n,
\end{equation}
where $n=V/4$ is the number of bubbles.

As mentioned above, the bicolored cycles of a jacket represent the faces of the corresponding discrete surface. Those bicolored cycles are also objects of the graph $G$ and we will still call them faces for $G$. In particular, we call a {\bf face of color $i$} of $G$ a bicolored cycle with colors $\{0,i\}$. A face has even length, with an equal number of edges of both its colors. The length of a face is defined by this number.

The degree of the graph which organizes the perturbative expansion of the model is given by the sum of the genera of the jackets,
\begin{equation}
\omega = g_1 + g_2 + g_3,
\label{deg_CTKT}
\end{equation}
from which we know that $\omega \geq 0$.
%Since a bubble has $4$ vertices and $6$ edges, it is convenient to obtain an expression of $\omega$ depending only on its number of bubbles $n$ and its number of faces. Also note that for $i,j\in\{1,2,3\}$ with $i \neq j$, the number of faces $F_{ij}$ of color $(i,j)$ is exactly one per bubble of the graph. Therefore we have for a connected graph with $n$ bubbles:
Using \eqref{Euler_charac} for all three colors, one finds
\begin{equation}
\omega = 3 + \frac{3}{2}n - \left(F_{1} + F_{2} + F_{3}\right).
\label{deg_CTKT2}
\end{equation}
%This formula indicates that ultimately, the degree only depends on number of faces of color $(0,i)$ for $i\in\{1,2,3\}$ and on the number of bubbles of the graph. Hence, these faces are simply referred to as faces of color $i$ in the following.

It was proved in~\cite{CTKT} that the graph of degree $0$ of this model were exactly the melonic graphs. %therefore the CTKT model is indeed a SYK-like tensor model.

%%%%%%%%%%%%%%%%%%%%%%
\subsection{The MO model} \label{subsec:1-2}
%%%%%%%%%%%%%%%%%%%%%%

{As already mentioned in the introduction}, the multi-orientable model was first introduced within a group field theory context~\cite{MO_original}. 
%But since the dominant graph of its expansion are also the melonic graphs~\cite{MO_expansion}~\cite{MO}, the MO model is also a SYK-like tensor model. 
A detailed review of this model can be found in~\cite{MO_review}.
This type of model was related to the SYK model in \cite{KT}.

This model has 
%is similar to the CTKT model, but with 
a complex Fermionic tensor field $\psi_{ijk}$. However
 the symmetry associated to this tensor is 
%not 
%$\mathcal{U}(N)^{3}$ as one could expect but  
${U}(N) \times {O}(N) \times {U}(N)$. This stems from the interaction term of this model which is:
\begin{equation}
I = \psi_{abc}\bar{\psi}_{ade}\psi_{fbe}\bar{\psi}_{fdc}
\label{MO_int}
\end{equation}
In this model, the interaction term is represented as a 4-valent vertex, where a 
{field} 
$\psi$ is an incident half-edge decorated with the sign $+$ and a 
{field} 
$\bar{\psi}$ is an incident half-edge with the sign $-$. To represent the indices, it is customary to blow up the half-edges into three strands, one for each index, and then connect the strands according to the contraction pattern,
\begin{equation}
I = \begin{array}{c} \includegraphics[scale=0.40]{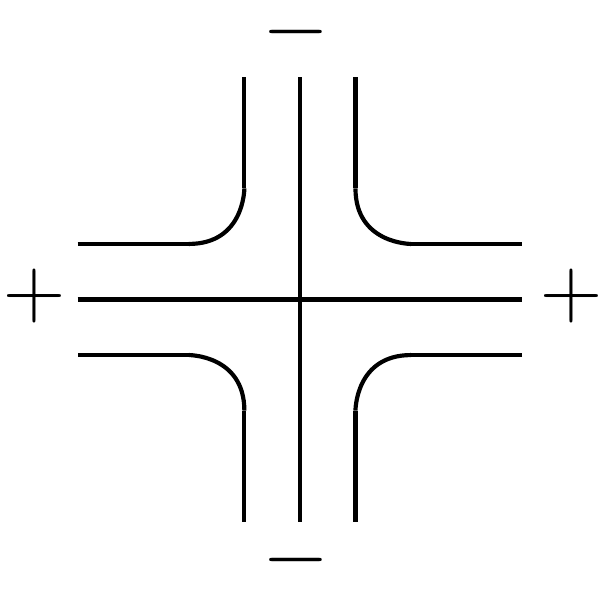} \end{array}
\end{equation}

%\begin{figure}[ht]
%\centering
%\includegraphics[scale=0.50]{Mo-Vertex.pdf}
%\caption{The vertex of the MO model, depicted with its strands and signs on the left, and its equivalent form as a 4-valent vertex with oriented incident edges on the right.}
%\label{MO_vertex}
%\end{figure}

The 
{$(0+1)-$dimensional}
SYK-like action 
{ of the MO model} 
writes:
\begin{equation}
S_{MO}=\int dt \left(\imath \bar  \psi_{abc}  \partial_t \psi_{abc} + \frac{\lambda}{2} \psi_{abc}\bar{\psi}_{ade}\psi_{fbe}\bar{\psi}_{fdc}\right).
\end{equation}

In the rest of this subsection we follow \cite{MO_expansion}
{ and \cite{Fusy2}}.
The propagator of this model is a 3-stranded edge which propagates each index of $\phi$, connecting two half-edges of different signs. When arriving onto a sign $+$ at a vertex, we call the {\bf left} strand the one which goes to the left, the {\bf right} strand the one which goes to the right and the {\bf straight} strand the one going straight. Those notions are preserved by the propagators so that it makes sense to define left, right and straight strands on the whole graph, denoted $L$, $R$, $S$ respectively.

%Since all the edges of the graph connects vertices of different signs, it is possible to give an orientation to these edges, say from negative half-edge to positive half-edge. Therefore, an MO graph can  be represented as an oriented 4-regular map given by its oriented $S$ strand. An example is represented on fig.~\ref{ex_MO_graph}.
%\begin{figure}[ht]
%\centering
%\includegraphics[scale=0.45]{fig4.eps}
%\caption{A Feynman graph of the MO model and its 4-regular map given by its oriented $S$ strand}
%\label{ex_MO_graph}
%\end{figure}

Denote 
{ by} 
$F$ the number of closed strands, ${F=}F_L + F_S + F_R$. The degree of the MO model for a connected graph with $n$ vertices is
\begin{equation}
\omega_{MO} = 3 + \frac{3}{2}n - F{.}
\label{deg_MO}
\end{equation}
Note that this formula is similar to the one obtained in the case of the CTKT model, except for the replacing of the faces of color 1, 2, 3 with closed strands of type $L$, $R$, $S$. 

In fact, the MO model also has a notion of jackets. The jacket $J_i$ is the 2-stranded graph obtained by deleting the type $i\in\{S,L,R\}$ strand. This leads to a similar formula as in the \ON-invariant cas, in term of the genus of the jackets{:}
\begin{equation}
\omega_{MO} = g_L + g_S + g_R.
\label{deg_MO2}
\end{equation}
Indeed, 
{ recall that} 
a 2-stranded graph, also called a fat graph or a ribbon graph, represents a discrete surface whose faces are the interior of the closed strands and whose genus is given by Euler formula.

%Since the strand $S$ is associated to an oriented 4-regular map, the edges of jacket $J_S$ can be oriented as well. Therefore, its graph can be embedded on an orientable surface. Hence, the jacket $J_S$ necessarily has integer genus $g_s$.

%%%%%%%%%%%%%%%%%%%%%%
\section{Relating MO graphs to \ON-invariant graphs} \label{sec:Relation}
%%%%%%%%%%%%%%%%%%%%

%In this section, we compare the graphs of the CTKT and MO models. First, we remind the construction that allows to associate a graph of the MO model to a unique graph of the CTKT model presented in ~\cite{CTKT}, then we give a reciprocal to this result to characterize  graphs of the CTKT model that cannot be associated to any graph of the MO model. All the graphs are supposed to be connected. If the graph is non-connected, we just have to apply this point to each of its connected component. 

In this section we first recall the construction from 
{ \citep{CTKT}} 
 allowing to associate an \ON-invariant graph to any MO graph. The reciprocal proposition is not true - counter-examples are easily found and the class of \ON-invariant graphs is strictly larger than the class of MO graphs. We nevertheless give a sufficient condition under which the reciprocal holds. It can be formulated topologically, as a jacket being orientable, or combinatorially, as the genus being an integer.

\begin{theorem} \label{thm:Bijection}
There is an explicit bijection between \ON-invariant graphs with a marked, orientable jacket and MO graphs supplemented with a color in $\{1,2,3\}$. 
{ This bijection} 
maps bubbles to vertices and faces to closed strands.
\end{theorem}

\begin{proof}
We first build a map from MO graphs to a subset of \ON-invariant graphs. To do so, we identify the interaction of the MO model with a bubble of the \ON-invariant model in the following way
\begin{equation} \label{MappingMO-CTKT}
\begin{array}{c} \includegraphics[scale=.5]{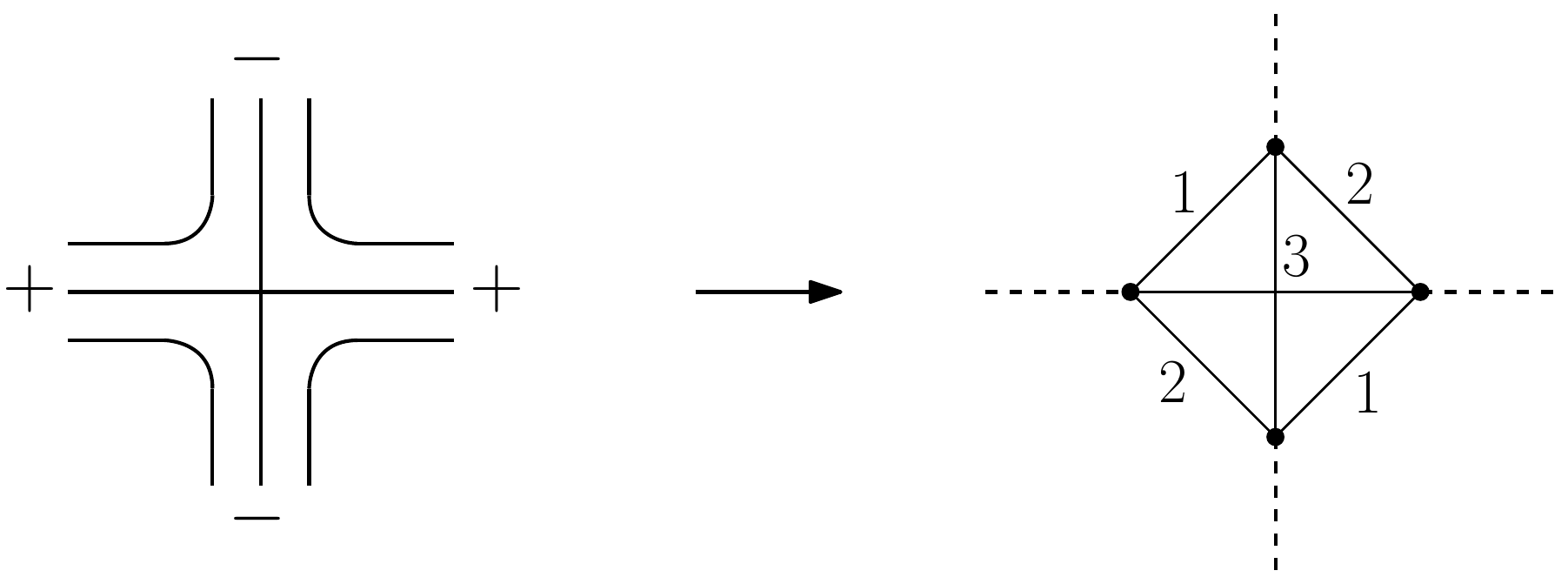} \end{array}
\end{equation}
which induces the following correspondence

{\centering
\begin{tabular}{|l|l|}
\hline
MO interaction & \ON-invariant bubble\\
\hline
Left strands & Edges of color 1\\
\hline
Right strands & Edges of color 2\\
\hline
Straight strands & Edges of color 3\\
\hline
Edges & Edges of color 0\\
\hline
\end{tabular}
\par}
Propagators then extends this correspondence to the whole graph, so that left strands become faces of color 1, right strands faces of color 2 and straight strands faces of color 3. If $G$ is an MO graph, then denote $\tilde{G}$ the corresponding \ON-invariant graph. Moreover if $i\in\{1,2,3\}$ is a color, define $\tilde{G}_i$ by exchanging the colors $i$ and 3 in $\tilde{G}$.

We now show that the jacket $J_i$ in $\tilde{G}_i$ is orientable. By definition, this is the same as $J_3$ being orientable in $\tilde{G}$. Due to the correspondence we have just described, this is in turn equivalent to the jacket $J_S$ being orientable in $G$. The latter statement is now easily proved, as $J_S$ in $G$ is a 2-stranded graph obtained by removing all straight strands. This means that it is represented as a ribbon graph with 4-valent ribbon vertices and ribbon edges which do not twist the ribbons, 
\begin{equation}
\begin{array}{c} \includegraphics[scale=.5]{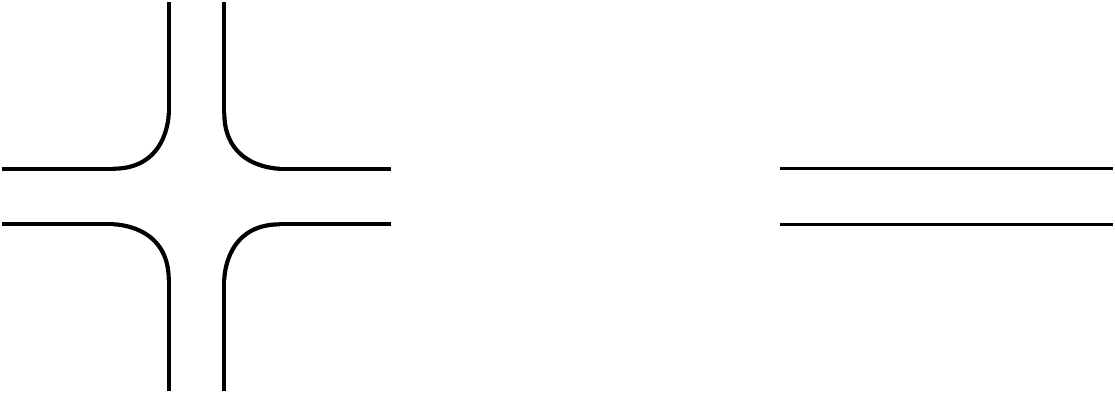} \end{array}
\end{equation}
These local rules generate orientable surfaces, as recalled in Theorem \ref{thm:OrientableRibbon}. This means that $J_i$ in $\tilde{G}_i$ is orientable.

Next we invert the map starting from $\tilde{G}_i$ an \ON-invariant graph with orientable, marked jacket $J_i$ for $i\in\{1, 2, 3\}$. We exchange the color $i$ with 3 and keep $i\in\{1,2,3\}$ as additional data for the MO graph $G$ that we are going to get. We thus consider $\tilde{G}$ an \ON-invariant graph whose jacket $J_3$ is orientable. The difficulty in inverting \eqref{MappingMO-CTKT} is to find how the signs $+/-$ can appear.

This is obviously due to orientability of $J_3$. The latter is an edge-colored graph, with colors $\{0,1,2\}$. As recalled in Theorem \ref{thm:OrientableGraph}, its orientability is equivalent to the edge-colored graph being bipartite. The jacket $J_3$ is thus bipartite. We color its vertices, say, black and white. Since the vertices of $G$ and $J_3$ are the same, one obtains a coloring of the vertices of $G$ ($G$ is not bipartite however). One can then apply the following mapping
\begin{equation}
\begin{array}{c} \includegraphics[scale=.5]{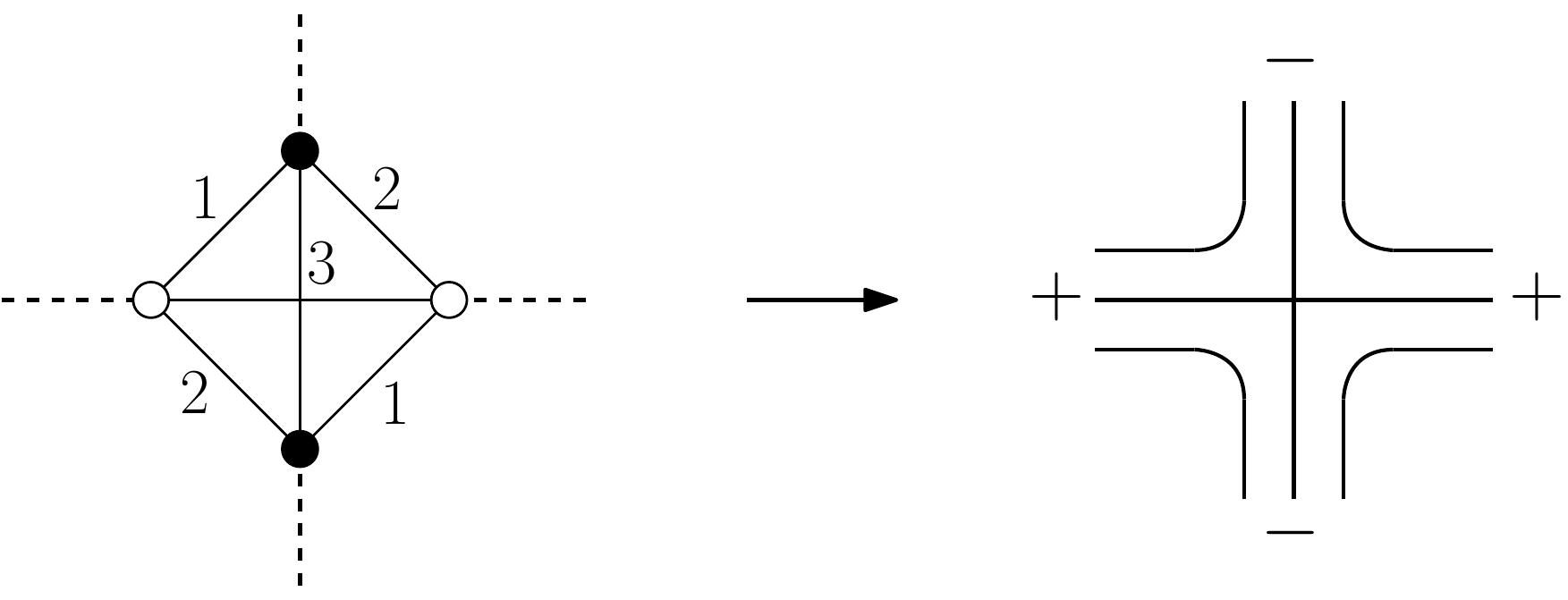} \end{array}
\end{equation}
In other words, white/black vertices inherited from orientability become the $+/-$ signs we were looking for. Edges of color 0 connect white to black vertices in $G$. This implies that $+$ can only connect to $-$ and the other way around. This way, $\tilde{G}$ is mapped to a{n} MO graph $G$, with the color $i$ of the marked jacket needing to be stored in 
addition with $G$.
\end{proof}

%%%%%%%%%%%%%%%%%%%%%
\section{Diagrammatic techniques for \ON-invariant graphs} \label{sec:tehnici}

In this section we develop explicit diagrammatic techniques which we use to study the graphs of degrees $1$ and $3/2$ of the \ON-invariant tensor model.

%%%%%%%%%%%%%%%%%%%
\subsection{2-edge-cuts}
\label{sec:EdgeCuts}

Consider $G$ a 2-particle-reducible (2PR) graph, i.e. a graph with a 2-edge-cut: a pair of edges $\{e, e'\}$ of color 0 whose removal disconnects $G$,
\begin{equation} \label{2EdgeCut}
G = \begin{array}{c} \includegraphics[scale=.4]{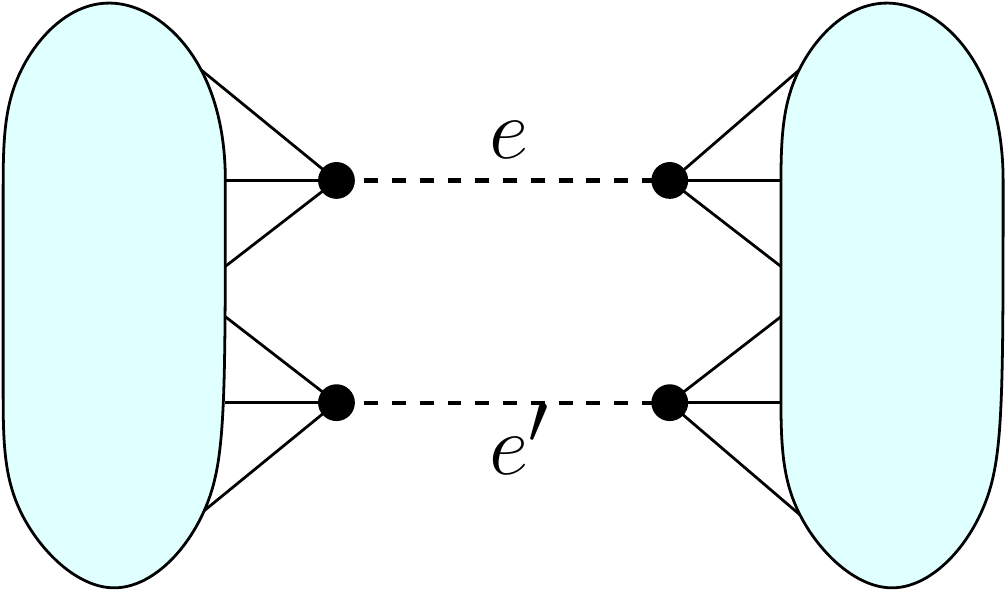} \end{array}
\end{equation}
There is a natural flip operation which turns $G$ into a pair of graphs,
\begin{equation} \label{2EdgeCutDisconnected}
G_L = \begin{array}{c} \includegraphics[scale=.4]{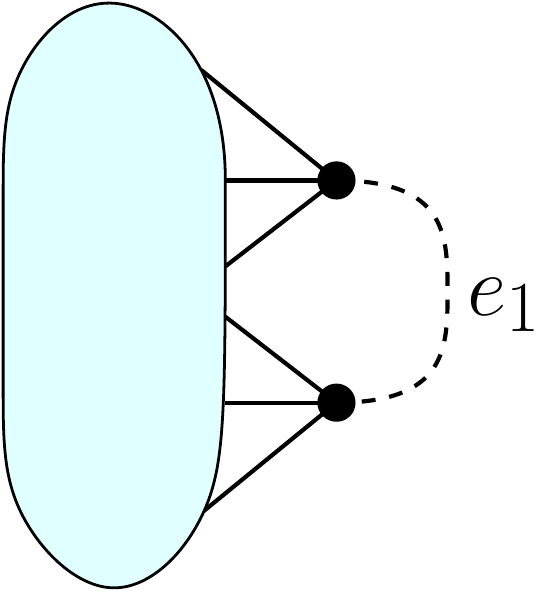} \end{array} \hspace{2cm} G_R = \begin{array}{c} \includegraphics[scale=.4]{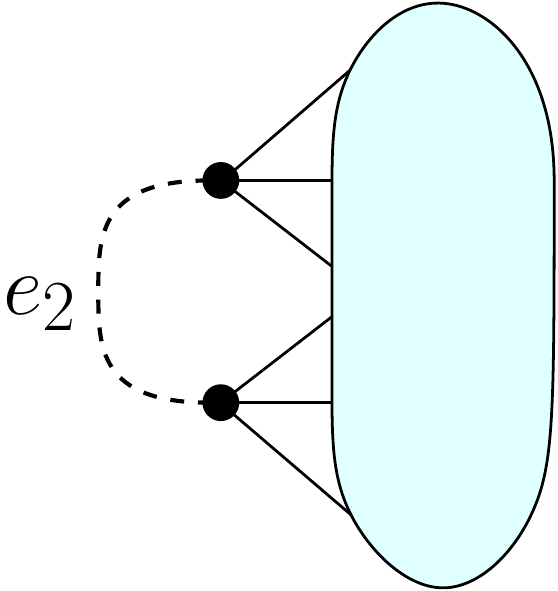} \end{array}
\end{equation}
by cutting $e, e'$ into half-edges and gluing them as above.

The following proposition is well-known.

\begin{proposition}
With the above notations
\begin{equation} \label{DegreeAddition}
\omega(G) = \omega(G_L) + \omega(G_R).
\end{equation}
This is the additivity of the degree for 2PR graphs.
\end{proposition}

\begin{proof}
Recall that each edge of color $0$ contributes to exactly one face of each color 1, 2, 3. 
{ Since $G$} 
 is 2PR, it is the same face of color $i$ which goes along $e$ and $e'$ for all $i=1,2,3$. However{,}  the{ re} are different faces along $e_1$ and $e_2$ since they live in different connected components. 
 %It comes
{ One has:}
\begin{equation}
F(G) = F(G_L) + F(G_R) - 3.
\end{equation}
The result 
{ then follows from} 
 the formula \eqref{deg_CTKT2} for the degree.
\end{proof}

Therefore, if one is interested in finding all graphs at a fixed value $\omega$ of the degree, it is possible to distinguish the cases of 2PR and 2PI graphs. The 2PR graphs of degree $\omega$ are given by 
\begin{itemize}
\item a 2PI graph of degree $\omega$ with insertions of 2-point graphs of vanishing degree ($\omega(G_R)=0$). Any such 2-point graph is obtained by recursive insertions of a fundamental graph called the \emph{melonic insertion},
\begin{equation} \label{MelonicInsertion}
\begin{array}{c} \includegraphics[scale=.5]{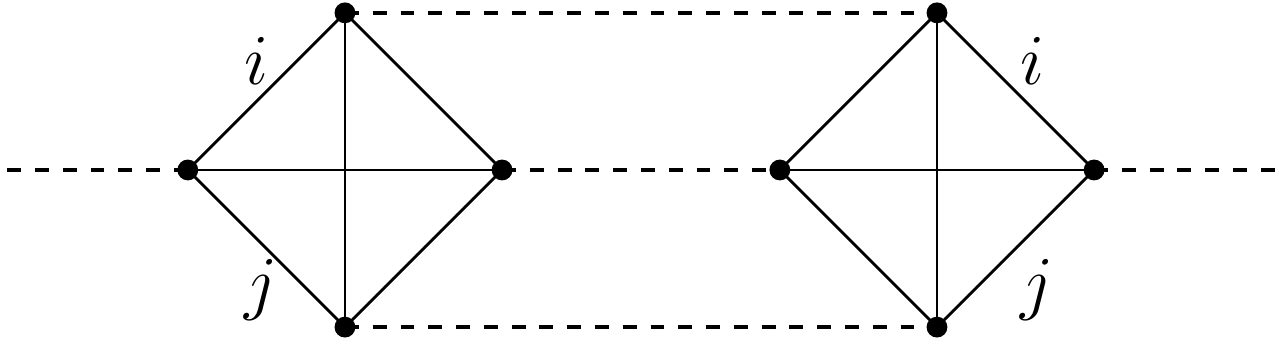} \end{array}
\end{equation}
to be inserted on any edge of color 0.
\item $G_L$ and $G_R$ both have degrees less than $\omega$.
\end{itemize}
%This leaves the 2PI graphs to be determined. To this effect we use the notion of dipoles.
%\begin{figure}
%\begin{center}
%\includegraphics[scale=0.6]{fig8.eps}
%\caption{An example of graph composition, the degree of the bottom graph is the sum of the degree of the two top graphs.}
%\label{degree_addition}
%\end{center}
%\end{figure}

%If $G_L$ is melonic, then $\omega(G) = \omega(G_R)$ and we can focus on analyzing $G_R$ instead of $G$. Similarly, if $G$ has a tadpole, i.e. a face of length $1$, then $\omega(G_L)=1/2$ and the tadpole can be removed 
%\begin{equation}
%\begin{array}{c} \includegraphics[scale=0.6]{fig9.eps} \end{array}
%\end{equation}
%while $\omega(G) = \omega(G_R) + 1/2$. Note that this is a straightforward adaptation of the tadpole removal used in \cite{Fusy1} for the MO model. As it is also well known from the tensor model literature, any melonic subgraph can also be deleted from the graph while preserving the degree (see, for example, \cite{CTKT}).
%
%Melonic and tadpole removals can thus be performed until $G$ has no pair of edges whose removal disconnects the graph and such that one connected component is melonic or a tadpole. They are called melon-free and tadpole-free graphs.

%%%%%%%%%%%%%%%%%%%%
\subsection{Dipole removals}
\label{sec:Dipoles}

Dipoles were introduced in~\cite{Fusy1} to facilitate the analysis of graphs at fixed degree in the MO model. Here we adapt the notion to the \ON-invariant model. Just 
{as} 
 in~\cite{Fusy1}, dipoles are defined as the minimal subgraphs which have an internal face of length two. There are no dipoles with a single bubbles. With two bubbles, 
 %it is easy to see \cite{Fusy1} that there are three types
{ one has three types of dipoles:}
\begin{equation}
\begin{array}{c} \includegraphics[scale=.5]{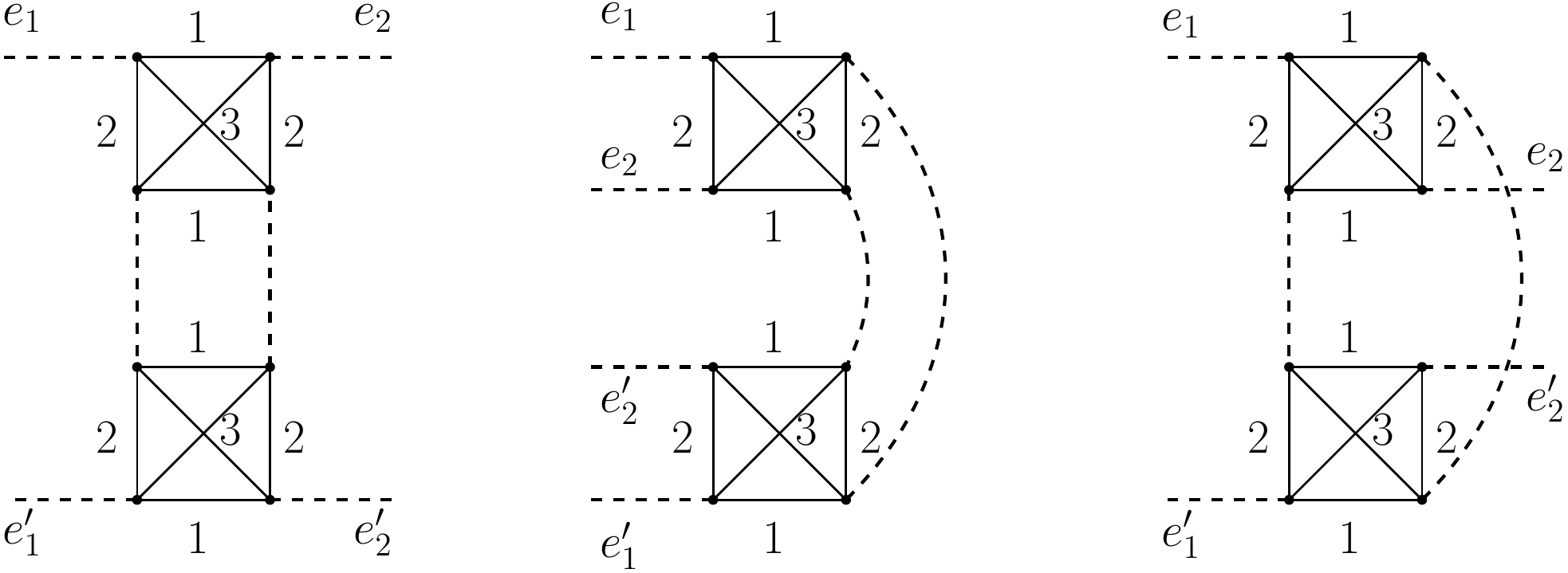} \end{array}
\end{equation}
We label them with the color of the their internal face of length 2: the left one has color 1, the middle one has color 2 and the right one has color 3.

If a 2PI graph contains a dipole, it can be removed while preserving connectedness\footnote{If $\{e_1, e_2\}$ (or $\{e'_1, e'_2\}$) forms a 2-edge-cut, then the dipole removal disconnects the graph.}
\begin{equation} \label{DipoleRemoval}
\begin{array}{c} \includegraphics[scale=.5]{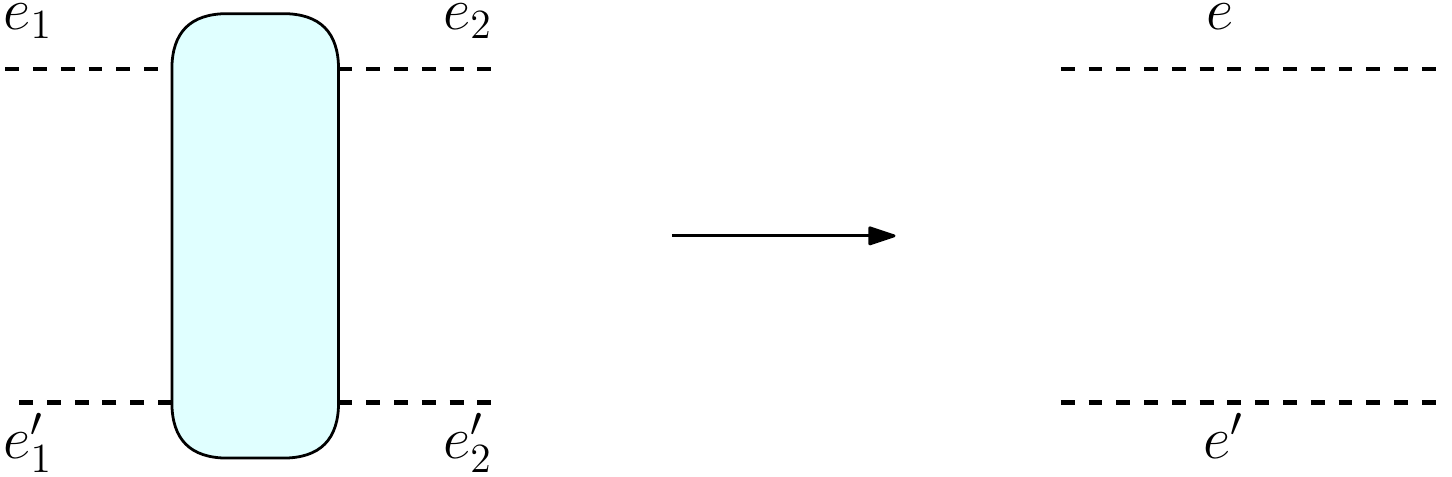} \end{array}
\end{equation}

\begin{proposition} \label{thm:Dipoles}
There are four types of dipole removals{, which are} described in the proof below. Three of them decrease the degree, by 4, 2 or 1, and the other one does not change it.
\end{proposition}

\begin{proof}
Let us calculate the variation of the degree. Without loss of generality, we consider the removal of a dipole of color 1. The new graph $G'$ has two bubbles less. It also loses a face of color 1. As for the faces of color 2 and 3, it depends on their paths in $G$. Up to symmetry, there are the following four possibilities
\begin{equation} \label{FaceCircuits}
\begin{gathered}
\begin{array}{c} \includegraphics[scale=.6]{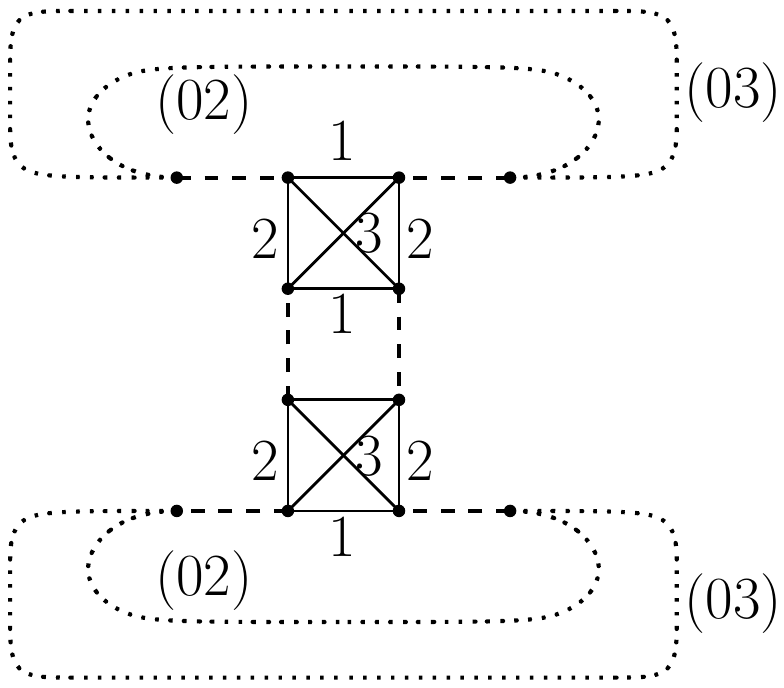} \end{array} \hspace{2cm} \begin{array}{c} \includegraphics[scale=.6]{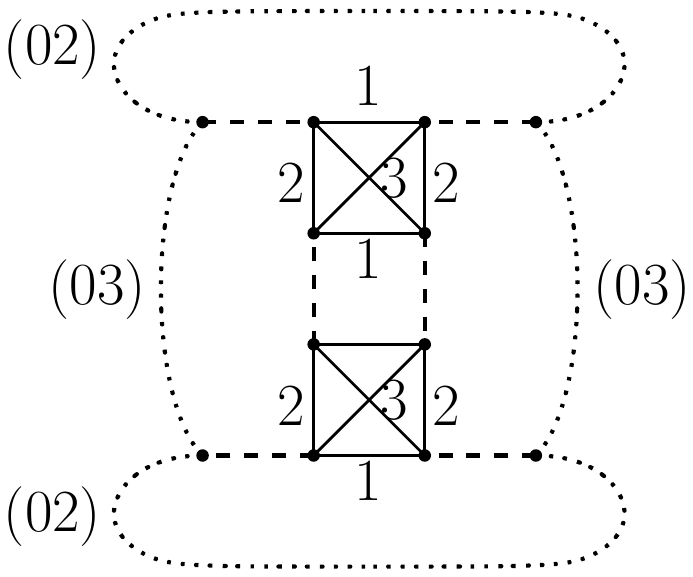} \end{array} \\
\begin{array}{c} \includegraphics[scale=.6]{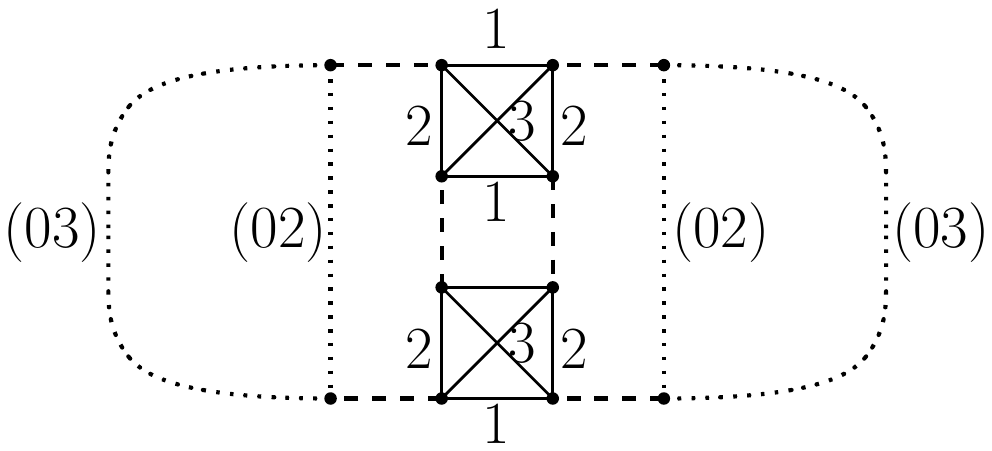} \end{array} \hspace{2cm} \begin{array}{c} \includegraphics[scale=.6]{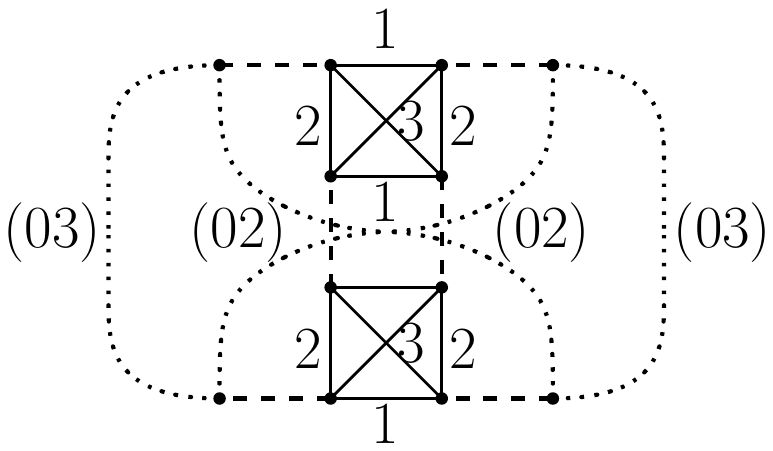} \end{array}
\end{gathered}
\end{equation}
Here the colored lines do not represent edges but instead the paths of the faces, i.e. bicolored cycles with alternating colors $0$ and $i$, which go along the bubbles of the dipole. Denote 
{by} 
$\epsilon_2, \epsilon_3 \in\{1,2\}$ the number of faces of color 2 and 3 in $G$ which go along the dipole, and 
{by} 
$\epsilon'_2, \epsilon'_3$ the number of those faces left in $G'$.
\begin{itemize}
\item In the first case (top left), $\epsilon_2=\epsilon_3=1$. Then each of those faces gets split into two\footnote{Notice that our drawing becomes disconnected after the removal, but this is misrepresentation due to the fact that our drawing does not contain other edges which would leave $G'$ connected.} in $G'$, i.e. $\epsilon'_2=\epsilon'_3=2$. 
%It comes that
{ One has:}
\begin{equation}
\begin{aligned}
\omega(G') &= 3+\frac{3}{2}(n(G)-2) - (F_1(G)-1 + F_2(G)+1 + F_3(G)+1)\\
&= \omega(G) - 4.
\end{aligned}
\end{equation}
In particular, since the degree is positive, it means that $\omega(G) \geq 4$.
\item In the second case (top right), $\epsilon_2=1$ and $\epsilon_3=2$. The face of color 2 is split into two in $G'$ while those of color 3 are merged, i.e. $\epsilon'_2=2$ and $\epsilon_3=1$. Therefore the number of faces of color 2 and 2 is globally unchanged and
\begin{equation}
\begin{aligned}
\omega(G') &= 3+\frac{3}{2}(n(G)-2) - (F_1(G)-1 + F_2(G)+1 + F_3(G)-1)\\
&= \omega(G) - 2.
\end{aligned}
\end{equation}
In particular, since the degree is positive, it means that $\omega(G) \geq 2$.
\item In the third case (bottom left), $\epsilon_2=\epsilon_3=2${,}  and for both colors the two faces are merged into one in $G'$, leaving $\epsilon'_2=\epsilon'_3=1$ and 
\begin{equation}
\begin{aligned}
\omega(G') &= 3+\frac{3}{2}(n(G)-2) - (F_1(G)-1 + F_2(G)-1 + F_3(G)-1)\\
&= \omega(G),
\end{aligned}
\end{equation}
i.e. the degree is unchanged by the dipole removal.
\item In the fourth case (bottom right), $\epsilon_2 = 1$ and $\epsilon_3=2$. However, in contrast with the previous cases, the face of color 2 is not split into two, i.e. $F_2(G)=F_2(G')$. The two faces of color 3 are merged into one, so $\epsilon'_2=\epsilon'_3=1$. This 
{ leads to} 
\begin{equation}
\begin{aligned}
\omega(G') &= 3+\frac{3}{2}(n(G)-2) - (F_1(G)-1 + F_2(G) + F_3(G)-1)\\
&= \omega(G)-1,
\end{aligned}
\end{equation}
\end{itemize}
This concludes the proof.
\end{proof}

%%%%%%%%%%%%%%%%%%%%
\subsection{Dipole insertions}

One can consider the reverse operation: the dipole insertion on the edges $e$ and $e'$. By inspecting the above proof, one finds the following lemma which provides a necessary condition for a dipole insertion not to increase the degree by more than 1.

\begin{lemma} \label{thm:NumberColors}
A dipole insertion on $\{e,e'\}$ which increases the degree by 1 or preserves it{,}  requires that there are two colors $\{i,j\} \subset \{1,2,3\}$ such that the face of color $i$ along $e$ also goes along $e'${,} and similarly for $j$.
\end{lemma}

\begin{proof}
A variation of the degree of at most 1 corresponds to the third and fourth cases of the above proof. After the dipole removals in those cases, the edges $e$ and $e'$ have two colors (2 and 3 in our pictures) for which the same face goes along $e$ and $e'$.
\end{proof}

%Then, there is only one situation in which the dipole insertion does not change the degree, corresponding to the reverse of the third case above. It is when one can find two edges of color 0 $e, e'$ with the same faces of color $i$ and $j$, for some $\{i,j\}\subset\{1,2,3\}$, going along $e$ and $e'$. This is not a sufficient criterion however since the fourth case above also satisfies it. It is necessary that the faces of color $i$ and $j$ go along $e$ and $e'$ ``in the same direction''.

We recall that melonic insertions do not change the degree of a graph. It is thus common, in order to analyze the degree, to remove at first all melonic 2-point functions, and then add the full, melonic 2-point function on all edges of color 0 at the end of the analysis. For instance, here we will aim at identifying all graphs of degree 1 so that we can try and identify all melon-free graphs fo{r} degree 1 and then add all melonic decorations on the edges of color 0. However, since we also consider dipole removals and insertions, the question arises as to whether new dipole insertions which change the degree by at most 1 would be allowed after melonic insertions. The following lemma answers this question.

\begin{lemma} \label{thm:Commuting}
Let $G'$ be obtained from $G$ by a melonic insertion on the edge $e$. Then there exists a possible dipole insertion preserving the degree on the pair $\{e',e''\}$ in $G$ iff it exists in $G'$, for $e',e''\neq e$.

%If $e'=e$ (or $e''=e$), then the dipole insertion preserving the degree is equivalent to 

\end{lemma}

{ This means that} melonic insertions do not ``offer'' new dipole insertions which preserve the degree or increase it by one.
%Dipole insertions which increase the degree by 1 or preserve it commute with melonic insertions.

\begin{proof}
Let $G$ be a graph and $G'$ obtained by a melonic insertion on $e$,
\begin{equation}
G = \begin{array}{c} \includegraphics[scale=.7]{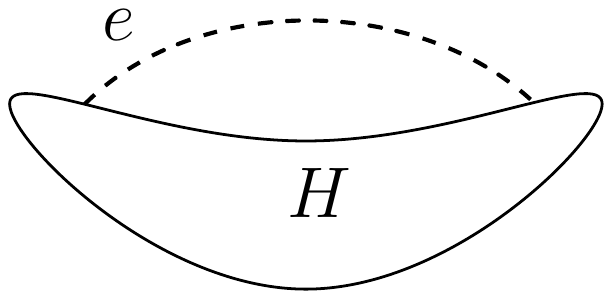} \end{array} \hspace{2cm}
G' = \begin{array}{c} \includegraphics[scale=.7]{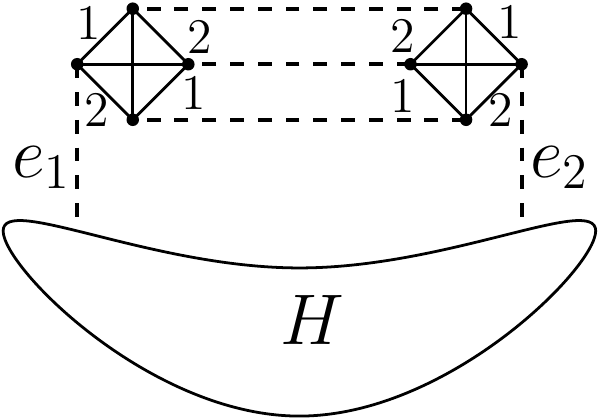} \end{array}
\end{equation}

To investigate the insertion of dipoles which increase the degree by one or preserve it, we use Lemma \ref{thm:NumberColors} and consider the pairs of edges of color 0 $\{e',e''\}$ such that there are two colors for which the faces going along $e'$ and $e''$ are the same. Also denote $\tilde{G}$ the graph obtained from $G$ by performing the dipole insertion on $\{e',e''\}$.

If $e'$ and $e''$ lie in $H$ in $G$, that is also the case in $G'$ so that the same dipole insertion can be made before or after the melonic insertion.

If $e'= e$ (or $e''=e$), then the corresponding dipole insertion gives rise to two possible insertions in $G'$, one with $e_1$ and $e''$, the other with $e_2$ and $e''$. In fact, the resulting graphs { only} differ by the position of the melonic insertion in $\tilde{G}$, so the melonic insertion does not provide new dipole insertions.

All the faces going along $e_1$ also go along $e_2$ so that it seems that a new dipole insertion is possible. However, because $e_1$ and $e_2$ share three faces instead of two, the dipole insertion is in fact a melonic insertion.

Finally, it is easy to check that for each of the three edges of color 0 connecting the bubbles of the melonic insertions, there are no other edges which would share at least two faces with it. Therefore, there is no possible dipole insertion which preserves the degree or change it by one.
\end{proof}

%%%%%%%%%%%%%%%%%
\subsection{Chains of dipoles}

A chain of dipoles is a sequence of dipoles of arbitrary length. 
{Such a chain} 
 can be of fixed color,
\begin{equation}
\begin{array}{c} \includegraphics[scale=.4]{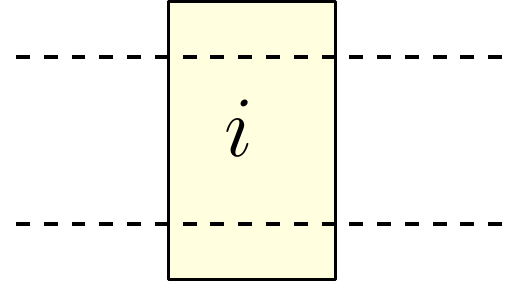} \end{array} = \begin{array}{c} \includegraphics[scale=.4]{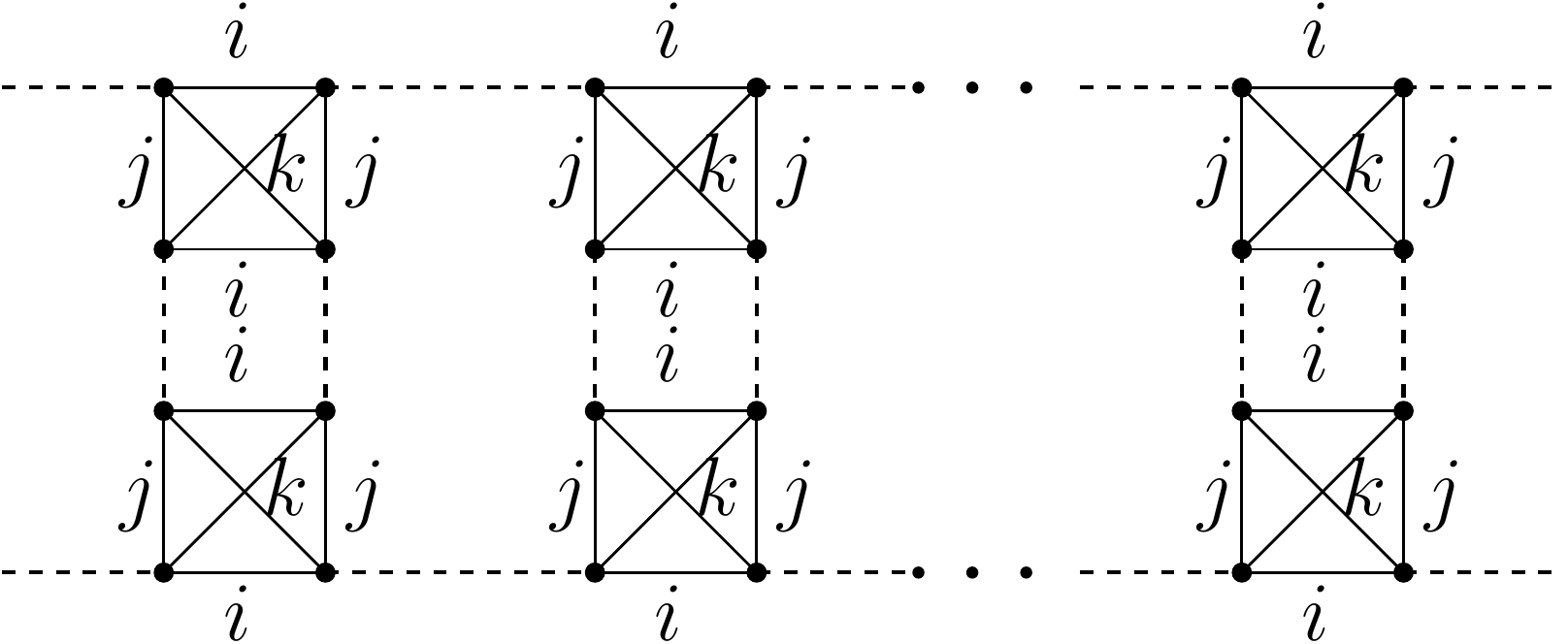} \end{array}
\end{equation}
or it can have multiple colors, in which case we use the same box diagram without any color label.

\begin{lemma}
\begin{enumerate}
\item If there is a pair $\{e,e'\}$ in $G$ where a dipole of color $i$ can be inserted without changing the degree, then a chain of color $i$ can be inserted without changing the degree.

\item If, in addition, a dipole of a different color can also be inserted without changing the degree, then any chain can be inserted without changing the degree.
\end{enumerate}
\end{lemma}

\begin{proof}
Without loss of generality, 
{let}  $i=1$.  
%for definiteness.
 Inserting a dipole without changing the degree means that we are in the third case of Proposition \ref{thm:Dipoles}, see \eqref{FaceCircuits}. There, the paths of the faces of color 2 and 3 had been drawn, but not those of color 1. There are two possible paths for the faces of color $i$ going along $e$ and $e'$: they are either different, or the same faces.
\begin{equation}
\begin{array}{c} \includegraphics[scale=.5]{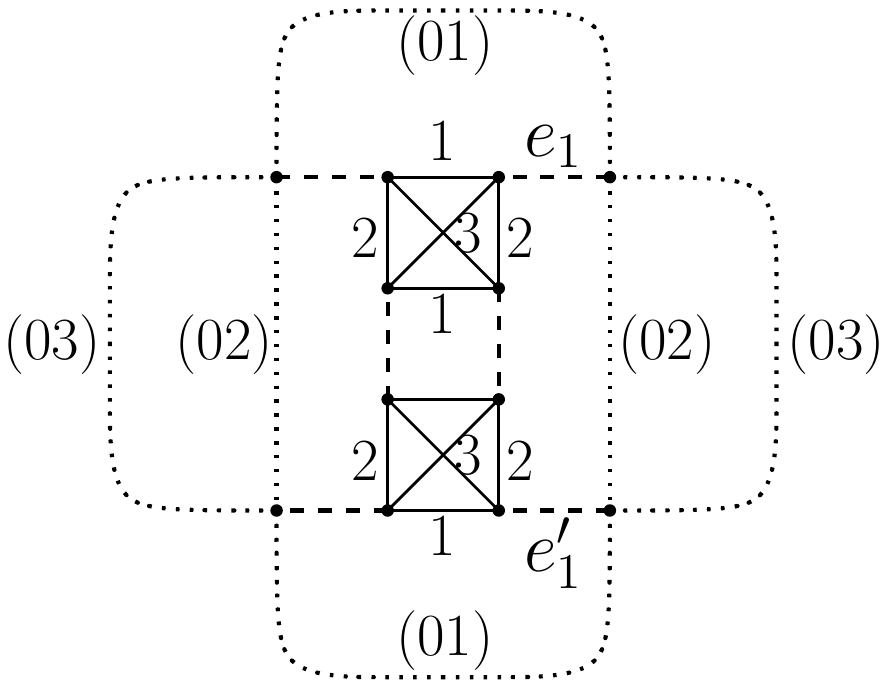} \end{array} \qquad \text{or} \qquad \begin{array}{c} \includegraphics[scale=.5]{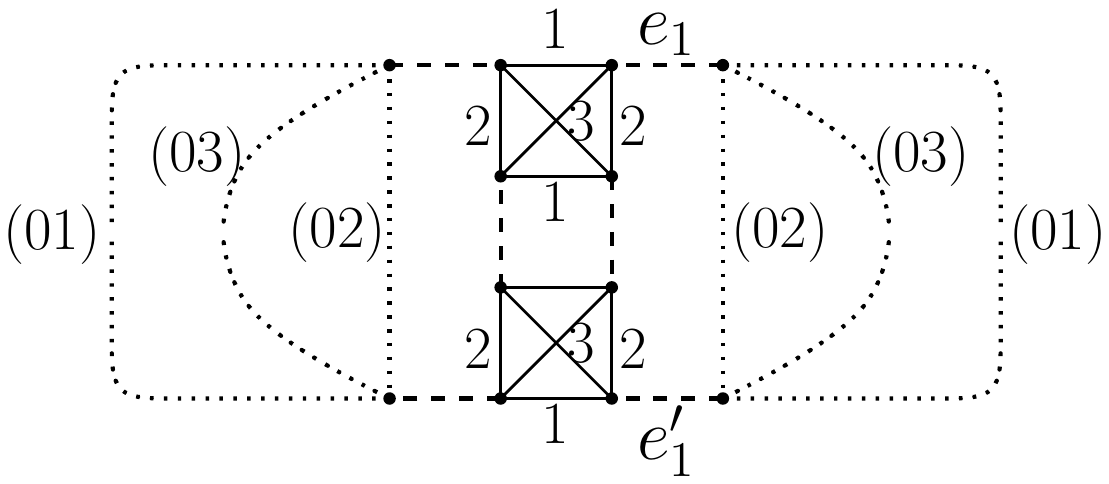} \end{array}
\end{equation}
In the first case, the edges $e_1$ and $e'_1$ share their faces of color 2 and 3, meaning that indeed a new dipole of color 1 can be inserted without changing the degree. However, a dipole of another color would change the degree because the faces of color 1 going along $e_1$ and $e'_1$ are different but would be merged after a dipole insertion of color 2 or 3.

In the second case, $e_1$ and $e'_1$ share their 3 faces so that any dipole insertions can be performed without changing the degree.
\end{proof}

Chains of dipoles are related to melonic graphs in the following way. If one opens up a melonic subgraph by cutting an edge of color 0, this creates a chain of dipoles. The following lemma is useful to glue two graphs along a 2-edge-cut as in Section \ref{sec:EdgeCuts}.

\begin{lemma} \label{thm:EdgeCutting}
If $G$ has the form
\begin{equation}
G = \begin{array}{c} \includegraphics[scale=.5]{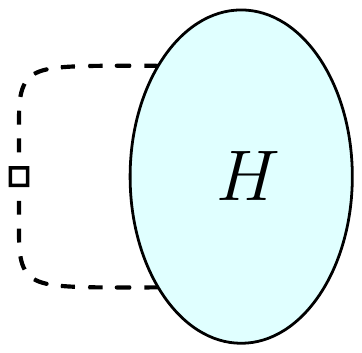} \end{array}
\end{equation}
where the square represents an arbitrary melonic 2-point function, then cutting an arbitrary edge of color 0 within this 2-point function results in a 2-point graph of the form
\begin{equation}
\begin{array}{c} \includegraphics[scale=.5]{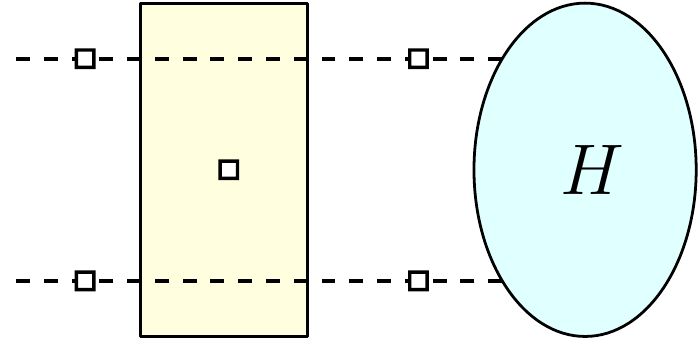} \end{array}
\end{equation}
where the square inside the chain means that all edges of color 0 can have melonic 2-point function {insertions}.
\end{lemma}

\begin{proof}
Let $e$ be the edge of color 0 which is cut. By definition of melonic 2-point functions, the most generic situation (up to color relabeling of each bubble) is like
\begin{equation}
G = \begin{array}{c} \includegraphics[scale=.5]{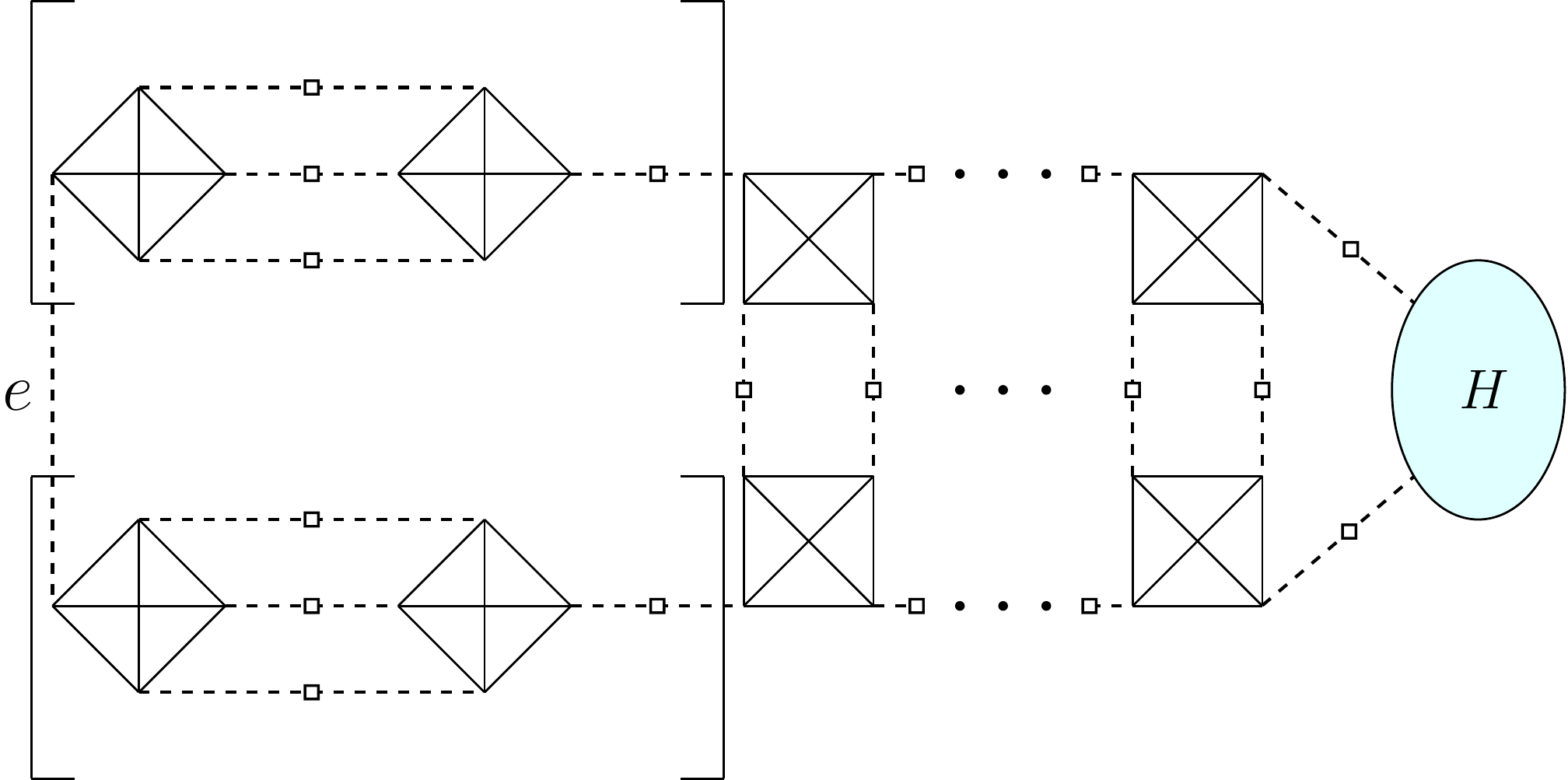} \end{array}
\end{equation}
where the parts between brackets may be empty in which case $e$ is extended and the three dots indicate potential dipole repetitions. Cutting 
{the edge} 
$e$ 
%then directly gives the lemma.
{leads to the expected result}.
\end{proof}

%%%%%%%%%%%%%%%%%%%%
\subsection{Face length}
\label{sec:FaceLength}

In this section we gather some simple results 
%about
{on}  
 the faces of \ON-invariant graphs.

\begin{lemma} \label{thm:FaceDegree}
If $G$ has degree $\omega$, $n$ bubbles and jackets $J_i$ of genus $g_i$, then
\begin{equation} \label{face_degree2}
F_i = 1 - \omega + 2 g_i +\frac{n}{2}.
\end{equation}
\end{lemma}

The degree is usually given 
{in the literature} 
in terms of the number of faces. Here we have an inverse, simple relation.

\begin{proof}
%Our notion of j
{Recall that} 
jackets 
{are defined} 
 such that there are as many jackets as colors. As it turns out, the linear system made of the three equations \eqref{Euler_charac}
 can be inverted to get an expression for the number of faces of a given color in terms of the Euler characteristics of the jackets,
\begin{equation} \label{face_degree}
F_i = \frac{1}{2}(\chi_j+\chi_k-\chi_i+n).
\end{equation}
Using the definition of the degree $\omega=\sum g_i$, 
%one arrives at 
{leads to identity}
\eqref{face_degree2},
{which concludes the proof}.
\end{proof}

A standard method in the analysis of colored graphs is to 
%control 
{analyze} 
the length of faces.
\begin{itemize}
\item A face of length one is equivalent to a tadpole, {\it i.e.} the subgraph
\begin{equation}
\begin{array}{c} \includegraphics[scale=.5]{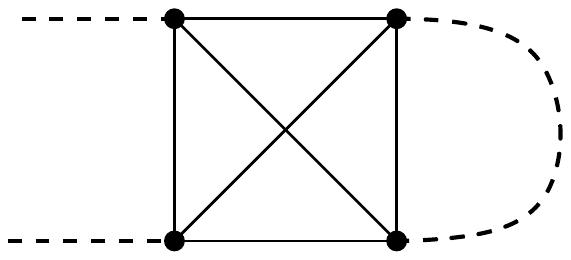} \end{array}
\end{equation}
for any color assignment. Then the graph is 2PR and after performing the flip as in \eqref{2EdgeCut}, \eqref{2EdgeCutDisconnected}, we have $G_R = \begin{array}{c} \includegraphics[scale=.2]{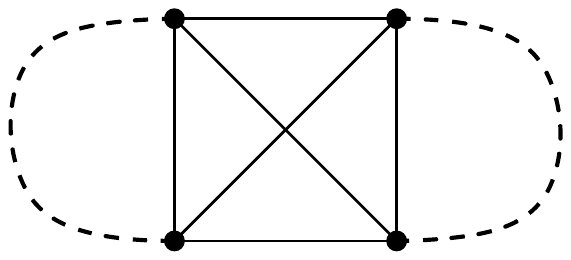} \end{array}$ which has degree $1/2$ and
\begin{equation}
\omega(G) = \omega(G_L) + 1/2.
\end{equation}
There are obviously no tadpoles in 2PI graphs.
\item A face of length 2 is part of dipole {(see the definition of dipoles above)} 
%, by definition of dipoles, which we studied above.
\end{itemize}
 We can thus focus on 2PI, dipole-free graphs, which in particular only have faces of length greater than or equal to three.

\begin{lemma} \label{FaceLength}
If $G$ is 2PI and dipole-free, then
\begin{equation} \label{cond_face}
\sum_{l\geq 5} (l-4)\,F_{i,l} = F_{i,3} + 4(\omega-2g_i-1)
\end{equation}
where $F_{i,l}$ is the number of faces of color $i$ and length $l$. 
%This equation is written with sums of positive terms only on both sides.
\end{lemma}

\begin{proof}
By noticing that each edge of color $0$ contributes to exactly one face of each color, and that there are $2n$ such edges, one obtains
\begin{equation}
\sum_{l\leq 1} l\,F_{i,l} = 2n.
\label{cond_long}
\end{equation}
Here $F_{i,l}$ is the number of faces of color $i$ and length $l$, and $n$ is the number of bubbles of the graph. Combining equations \eqref{face_degree2} and \eqref{cond_long}, the number of bubbles can be eliminated,
\begin{equation}
\sum_{l\leq 1} (l-4)\,F_{i,l} =  4(\omega - 2g_i -1).
\end{equation}
Restricting this equation to 2PI, dipole-free graphs implies that $l\geq 3$. The term $l=3$ in the sum is in fact the only 
 {term} 
%coming 
with a negative sign in the left hand side and we  extract it to move it to the right hand side.
\end{proof}

If a jacket is orientable, 
%another property is easily found.
{one has the following property:}

\begin{lemma} \label{thm:EvenLength}
If $G$ has its jacket $J_i$ orientable, then its faces of color $i$ are of even length, for $i\in\{1,2,3\}$.
\end{lemma}

\begin{proof}
Without loss of generality, assume  $i=3$. As already 
{noticed} 
in the proof of Theorem \ref{thm:Bijection}, when $J_3$ is orientable, there is a coloring of the vertices of $G$ such that $J_3$ is bipartite. $G$ itself is not bipartite since the edges of color 3 connect white vertices to white vertices and black vertices to black vertices,
\begin{equation} \label{VertexColoringCTKT}
\begin{array}{c} \includegraphics[scale=.5]{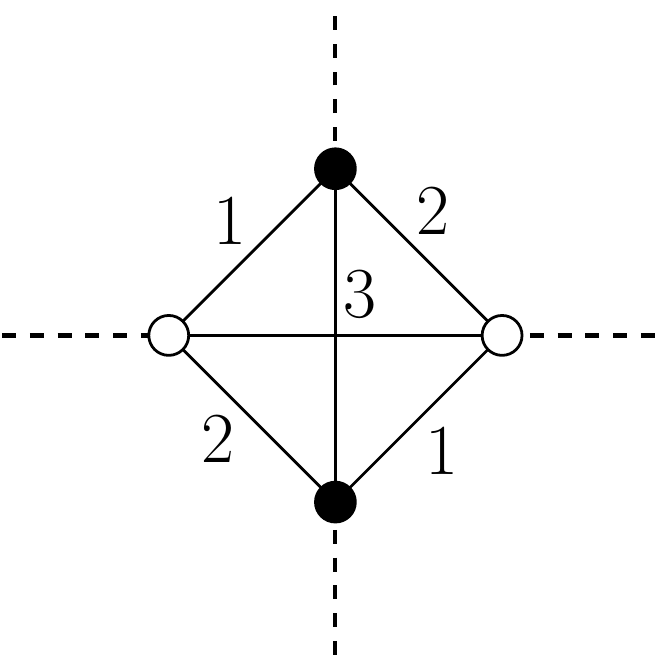} \end{array}
\end{equation}
The edges of color 0 (here dashed) then connect white to black and black to white vertices. We conclude that the faces of color 3 have even lengths.
\end{proof}

With the different operations described above, we will list in the following sections all the graphs of degree $1$ and $3/2$. For each degree, we first give the graphs that can be obtained through composition of smaller degree graph{s}. Then we give all the $2$PI graph{s} with no melons nor tadpoles, and finally we study the possible dipole insertions.  Our proof relies on 
%the knowledge of 
degree $0$ and degree $1/2$ \ON invariant graphs.

%%%%%%%%%%%%%%%%%%%%
\subsection{The strategy}
\label{sec:Strategy}

To identify all graphs at a given degree $\omega$, we propose the following strategy.
\begin{enumerate}
\item \label{enum:melon} Consider all graphs up to melonic 2-point functions on edges of color 0, since they leave the degree invariant.
\item \label{enum:2PR} Put aside the 2PR graphs of degree $\omega$. They are indeed obtained by composing two graphs $G_L, G_R$ of smaller degree and such that $\omega(G_L) + \omega(G_R) = \omega$. This leaves the 2PI graphs to be determined.
\item \label{enum:2PIDipoleFree} Find all 2PI, dipole-free graphs of degree $\omega$.
\item \label{enum:DipoleInsertions} Consider all (chains of) dipole insertions which preserve the degree on the graphs found above of degree $\omega$, and (chains of) dipole insertions on graphs of smaller degrees which bring it up to $\omega$. Then consider dipole insertions on the newly found graphs, and so on iteratively.
\item \label{enum:Induction} {Finally,} an induction on the number of bubbles allows to conclude. This 
{last step is} 
%appears 
necessary because of the iterative nature of the previous step.
\end{enumerate}
Notice that the steps \ref{enum:2PR} and \ref{enum:DipoleInsertions} are ``automatic'': they rely on simple operations (composition forming a 2-edge-cut, and dipole insertions) to be performed on graphs already known and can be in practice automated.
Only step \ref{enum:2PIDipoleFree} requires an independent analysis, which has to be performed {"}by hand{"}.

%%%%%%%%%%%%%%%%%
\section{Degree $0$ and $1/2$ {graphs} of the \ON invariant tensor model} \label{sec:Degree0}
%%%%%%%%%%%

%Here we just 
%{In this section} we
%recall the leading and next-to-leading order {graphs}, which have been studied in detail in \cite{CTKT}. 
%{These Feynman graphs} 
% correspond to degree 0 and $1/2$ {graphs}.

The degree $0$ and respectively $1/2$ { graphs of the \ON invariant model} are the melonic graphs and tadpoles {(see \cite{CTKT} for details)}. Melonic graphs are obtained by recursive melonic insertions, i.e. insertions of the 2-point graph \eqref{MelonicInsertion}, starting from the 2-bubble graph 
\begin{equation} \label{TwoBubbleGraph}
\begin{array}{c} \includegraphics[scale=.5]{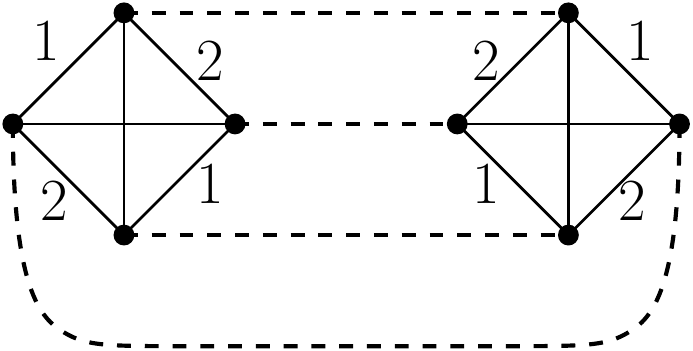} \end{array}
\end{equation}
which is invariant under color permutations.

The graphs of degree $1/2$ are obtained by recursive insertions of the same 2-point graph \eqref{MelonicInsertion} starting with the double-tadpole
\begin{equation} \label{DoubleTadpole}
\begin{array}{c} \includegraphics[scale=.5]{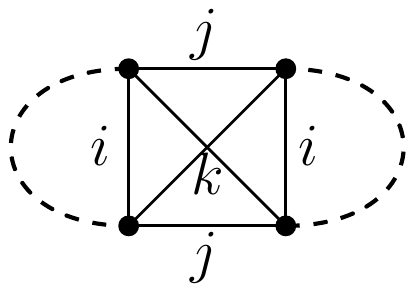} \end{array}
\end{equation}
Notice that there are 3 different double-tadpoles, depending on the color $i$ of the two faces of length one.
%All the graphs are supposed to be melon free, since inserting a melon on an edge is always possible and does not change the degree of the graph. 

%\begin{figure}[ht]
%\centering
%\includegraphics[scale=0.75]{fig14.eps}
%\caption{On the left, a melonic graph ($\omega =0$); on the right, a double tadpole graph ($\omega =1/2$)}
%\label{fig:MelonTadpole}
%\end{figure}

%%%%%%%%%%%%%%%%%%
\section{Degree $1$ graphs of the $O(N)$ invariant SYK-like tensor model} \label{sec:degre1}
%%%%%%%%%%%%%%%%%%

In this section we focus on {the} $\omega=1$ {case} and apply the strategy proposed in {subs}ection \ref{sec:Strategy}. We {thus} find all the Feynman graphs of degree $1$ of the \ON-invariant SYK-like tensor model. Notice that this also gives all the graphs of the MO model.

As already emphasized, the only non-automatic step of our strategy is Step \ref{enum:2PIDipoleFree} whose goal is to obtain the 2PI, dipole-free graphs of degree 1. This is what we start with, in {subs}ection \ref{sec:2PIDipoleFreeDegree1}. We then go directly to Step \ref{enum:Induction} since all the other steps are automatic, 
in {subs}ection \ref{sec:Degree1AllGraphs}.

%%%%%%%%%%%%%%%%%%%
\subsection{2PI, dipole-free graph of degree 1} \label{sec:2PIDipoleFreeDegree1}
%%%%%%%%%%%%%%%%%%%

\begin{theorem} \label{thm:Degree1}
There is an unique 2PI, dipole-free, \ON-invariant graph of degree 1, given below
\begin{equation} \label{deg12PI}
\begin{array}{c} \includegraphics[scale=0.5]{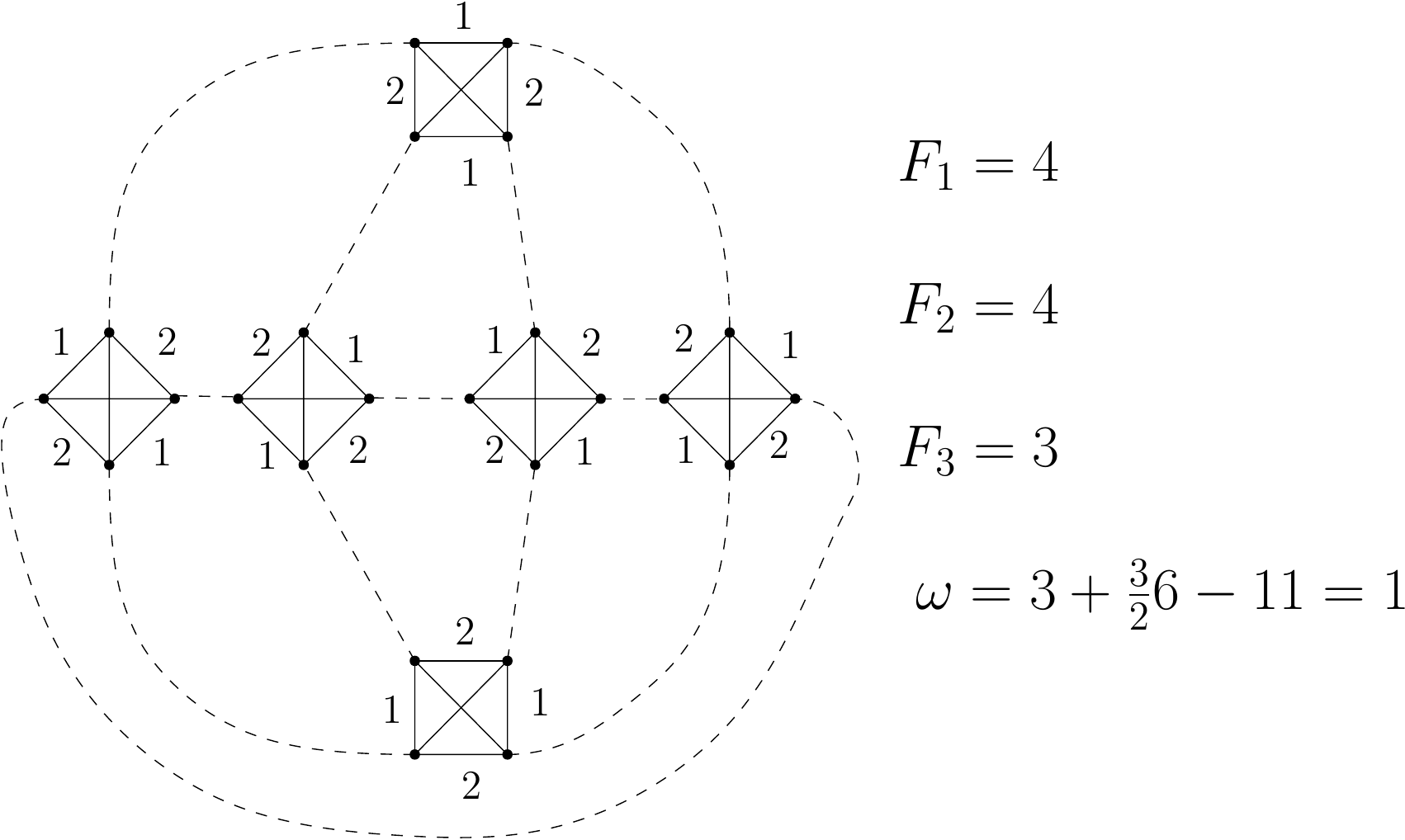} \end{array}
\end{equation}
\end{theorem}

\begin{proof}
Let $G$ be a 2PI graph of degree $1$ with no dipoles
{and} 
with $n$ bubbles. Considering \eqref{deg_CTKT}, the genera of its three jackets are either $(g_1,g_2,g_3)=(0,\frac{1}{2},\frac{1}{2})$ or $(g_1,g_2,g_3)=(0,0,1)$ up to color permutations. In both cases, $G$ has a planar jacket which we assume{,} without loss of generality{,} to be $J_3$, hence $g_3 = 0$.

Since $J_3$ is planar, it is bipartite and thus admits a canonical embedding described in Remark \ref{thm:CanonicalEmbedding}, obtained by using the counter-clockwise orientation for the colors $(012)$ around white vertices{,} and clockwise around black vertices. This in turn provides a canonical embedding for $G$ itself, obtained by adding the edges of color 3 at the corners between the colors 1 and 2. We use this representation in the remaining of the proof.

Moreover, among the two other jackets, at least one, say $J_1$, has non-zero genus, i.e. $g_1>0$. Using equation \eqref{cond_face} with $\omega=1$, we find{:}
\begin{equation}
F_{1,3} = 8 g_1 + \sum_{l\geq 5} (l-4)\, F_{1,l} > 0.
\end{equation}
Let us thus focus on a face of color 1 in $G$,
\begin{equation} \label{FaceLength3}
\begin{array}{c} \includegraphics[scale=.5]{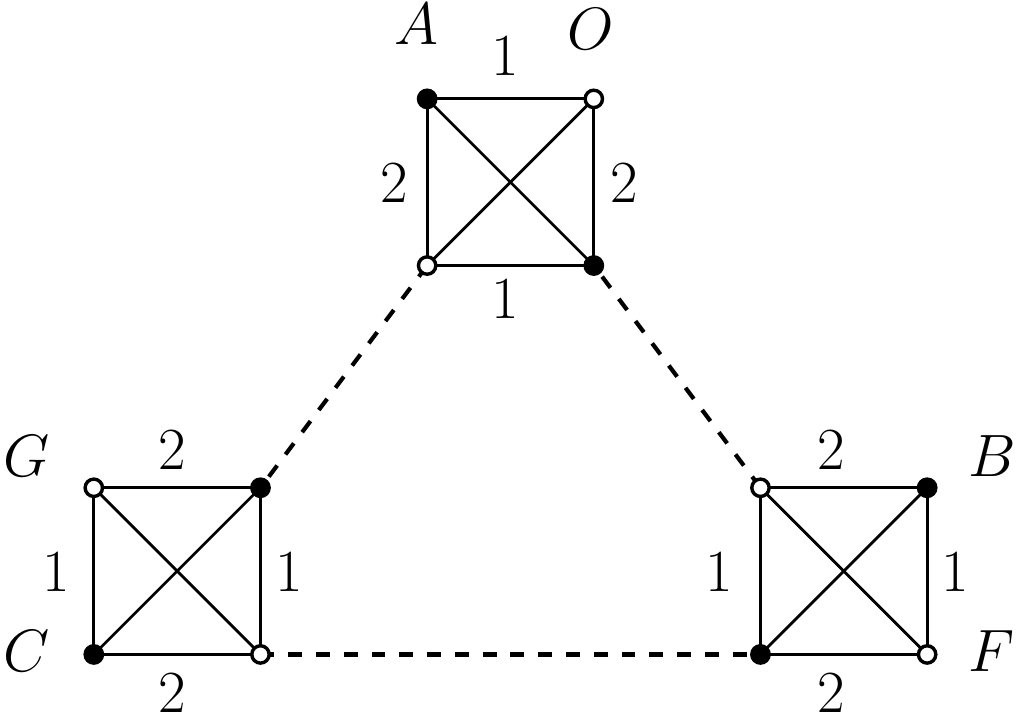} \end{array}
\end{equation}
where the letters label vertices.

To complete the graph, we need to investigate the faces of color 3. Lemma \ref{thm:EvenLength} applies to $J_3$: the faces of $G$ of color 3 are of even lengths. Since $G$ is 2PI, it has no tadpoles, so no faces of length 1, and since it is dipole-free, it also has no faces of length 2. The faces of color 3 thus have length at least 4. Equation \eqref{cond_face} with $g_3=0$ and $\omega=1$ becomes
\begin{equation}
\sum_{l\geq 5} (l-4)\,F_{3,l} = 0,
\end{equation}
meaning {that} all the faces of color $3$ of $G$ have to be of length exactly $4$. We will thus complete $G$ so that all its faces of color 3 have length 4.

Consider the face of color 3 going through the vertex $O$ and 
{recall} 
 that edges of color 0 only connect black to white vertices. If $O$ is connected to $A$, a tadpole would be created which is forbidden. If $O$ is connected $B$, a dipole would be formed{,} which is 
{also}  
 forbidden, and similarly for $C$. 
 {The vertex} $O$ must therefore be connected to a vertex of another bubble
\begin{equation}
\begin{array}{c} \includegraphics[scale=.5]{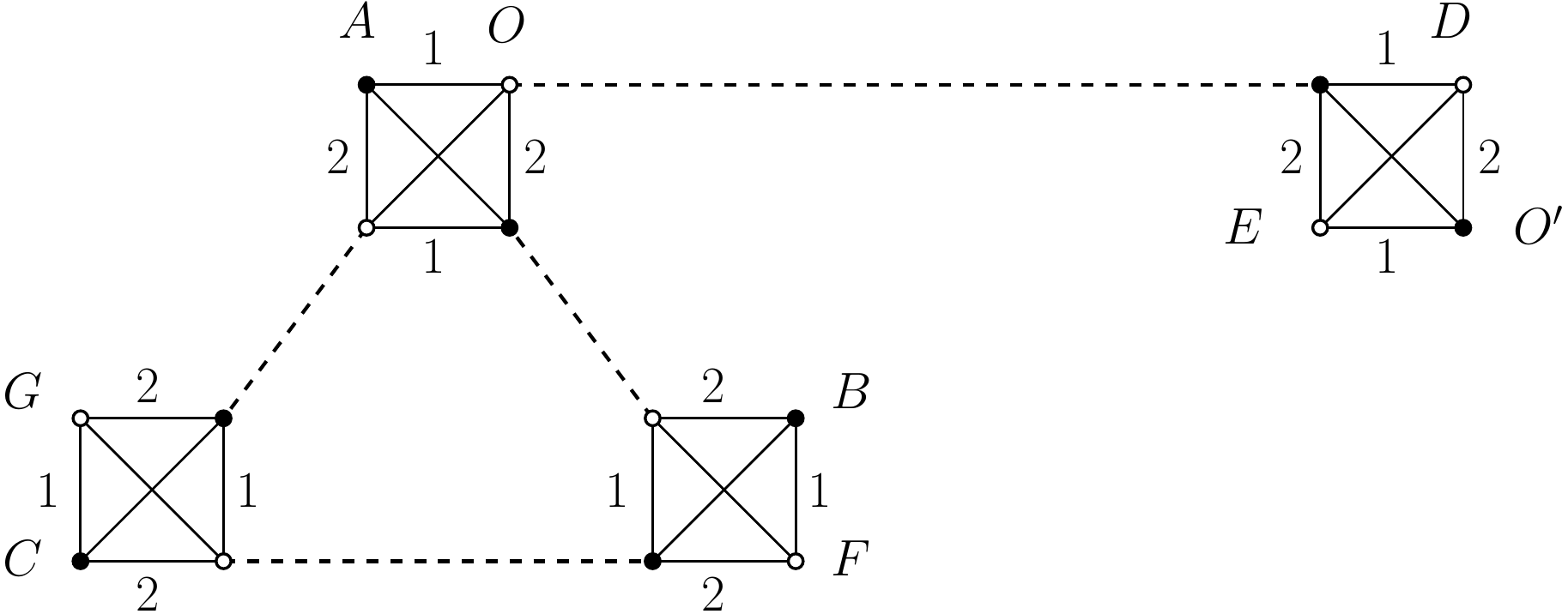} \end{array}
\end{equation}
The vertex $O'$ cannot be connected to $D$ or $E$ because this would create tadpoles. It cannot connect to $F$ or $G$ because the face would have length greater than 4. Therefore $O'$ must be connected to another bubble 
\begin{equation}
\begin{array}{c} \includegraphics[scale=.5]{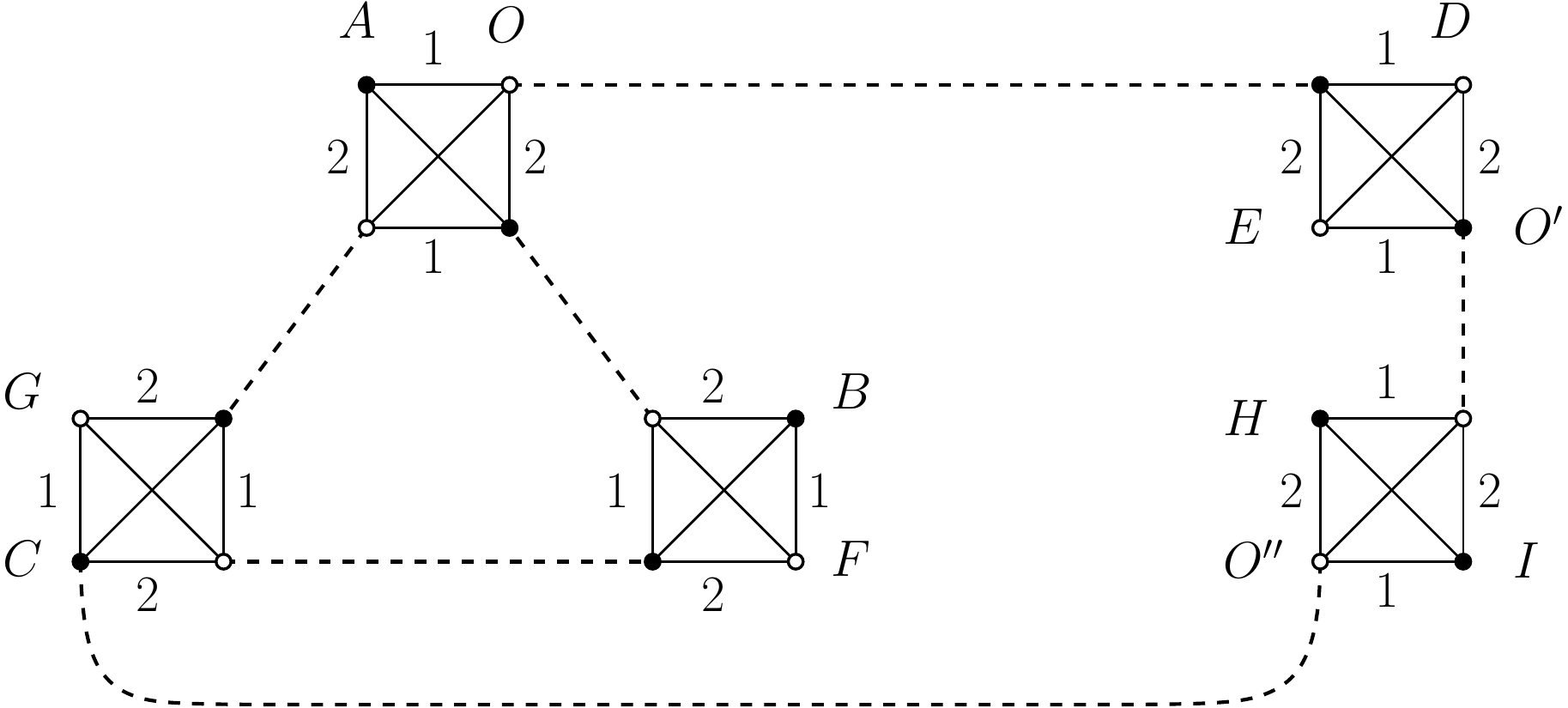} \end{array}
\end{equation}
where the connection of $O''$ to $C$ is forced to get a face of length 4.

Let us now consider the face of color 3 going through $A$ and $F$. First, assume that face connects to a new bubble and let us prove that this cannot be true. Planarity of $J_3$ requires the new bubble to lie in the same region as follows
\begin{equation}
\begin{array}{c} \includegraphics[scale=.5]{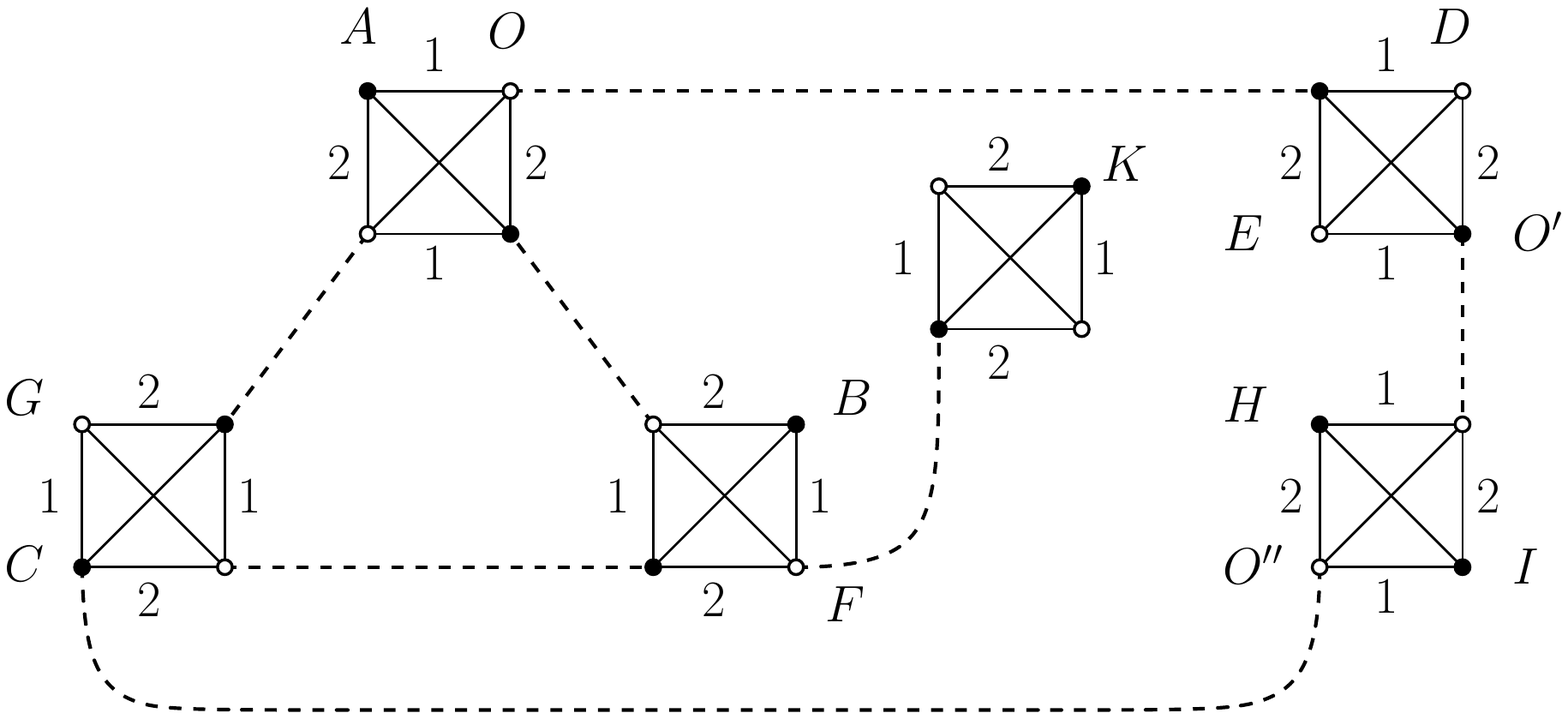} \end{array}
\end{equation}
Again $K$ cannot be connected to the white vertices of its bubbles as that would create tadpoles. If it is connected to $E$, this then forces an edge of color 0 from $D$ to $A$ so that the face be of length 4, but this creates a dipole. Any other connection would be non-planar for $J_3$, except if $K$ is connected to another new bubble. However, it is then easy to check that this new bubble would have to be connected to $A$ to close the face at length 4 and that would be non-planar.

The only possibility is thus that $F$ is connected to $H$ instead. Then {the vertex} $I$ cannot however be connected to {the vertices} $D$ or $G$ without creating a dipole. This requires $I$ to be connected to a new bubble,
\begin{equation}
\begin{array}{c} \includegraphics[scale=.5]{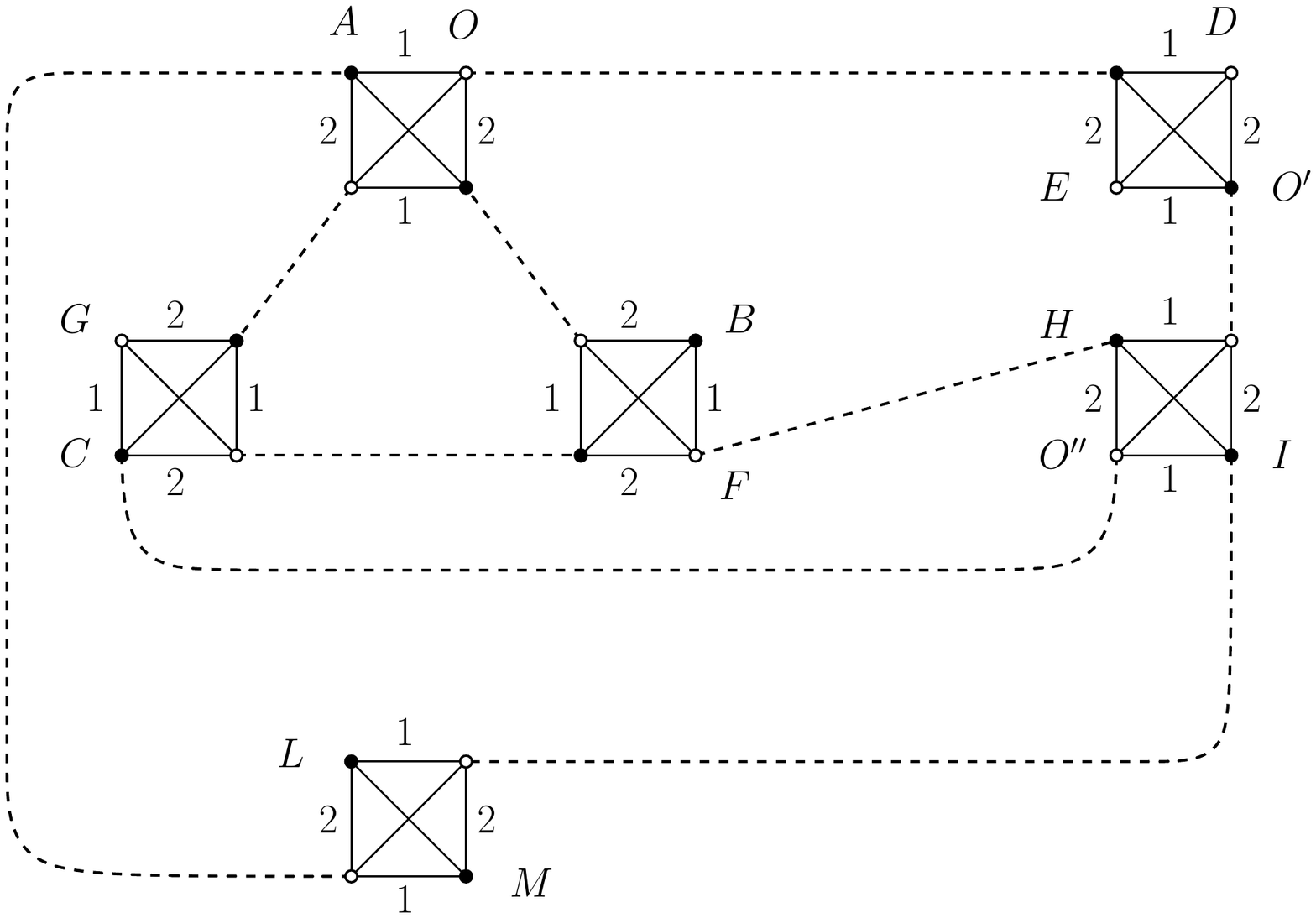} \end{array}
\end{equation}
which is in turn connected to $A$ to close the face of color 3 at length 4.

Finally, consider the face of color 3 going through $G$ and $B$. We explained above that $F$ could no be connected to a new bubble. The same reasoning applies to $B$, as can be checked directly, so that it has to be connected to $E$. Then we must create a path with two edges of color 0 and one edge of color 3 between $D$ and $G$. It is straightforward to check that if $D$ connects to a new bubble then one has to break planarity of $J_3$ to close the face. Eventually, the only possibility is to connect $D$ to $M$ and $L$ to $G$. This is the graph given in the theorem.
\end{proof}

%%%%%%%%%%%%%
\subsection{The graphs of degree 1} \label{sec:Degree1AllGraphs}
%%%%%%%%%%%%%

\begin{theorem} \label{thm:Degree1AllGraphs}
The graphs of degree 1 are the graph given in \eqref{deg12PI} and the following graphs
\begin{gather}
\begin{array}{c} \includegraphics[scale=.4]{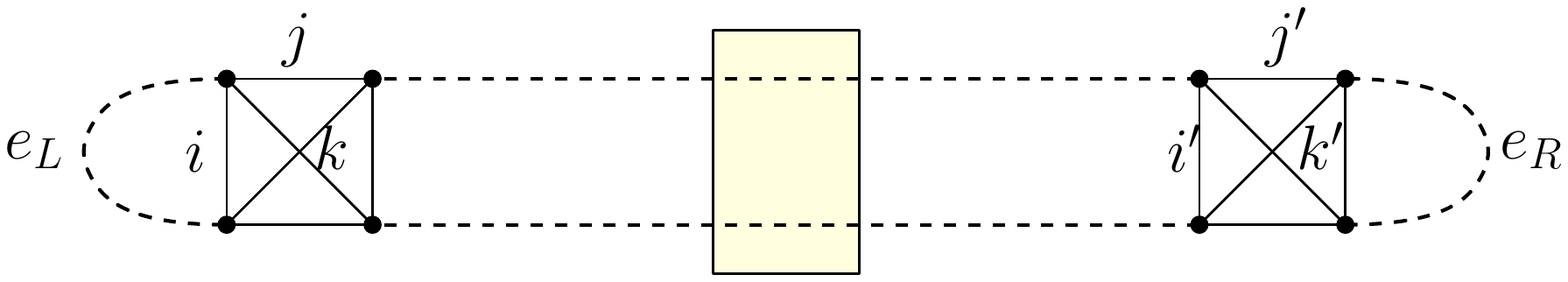} \end{array} \label{DoubleTadpoleChain} \\
\begin{array}{c} \includegraphics[scale=.4]{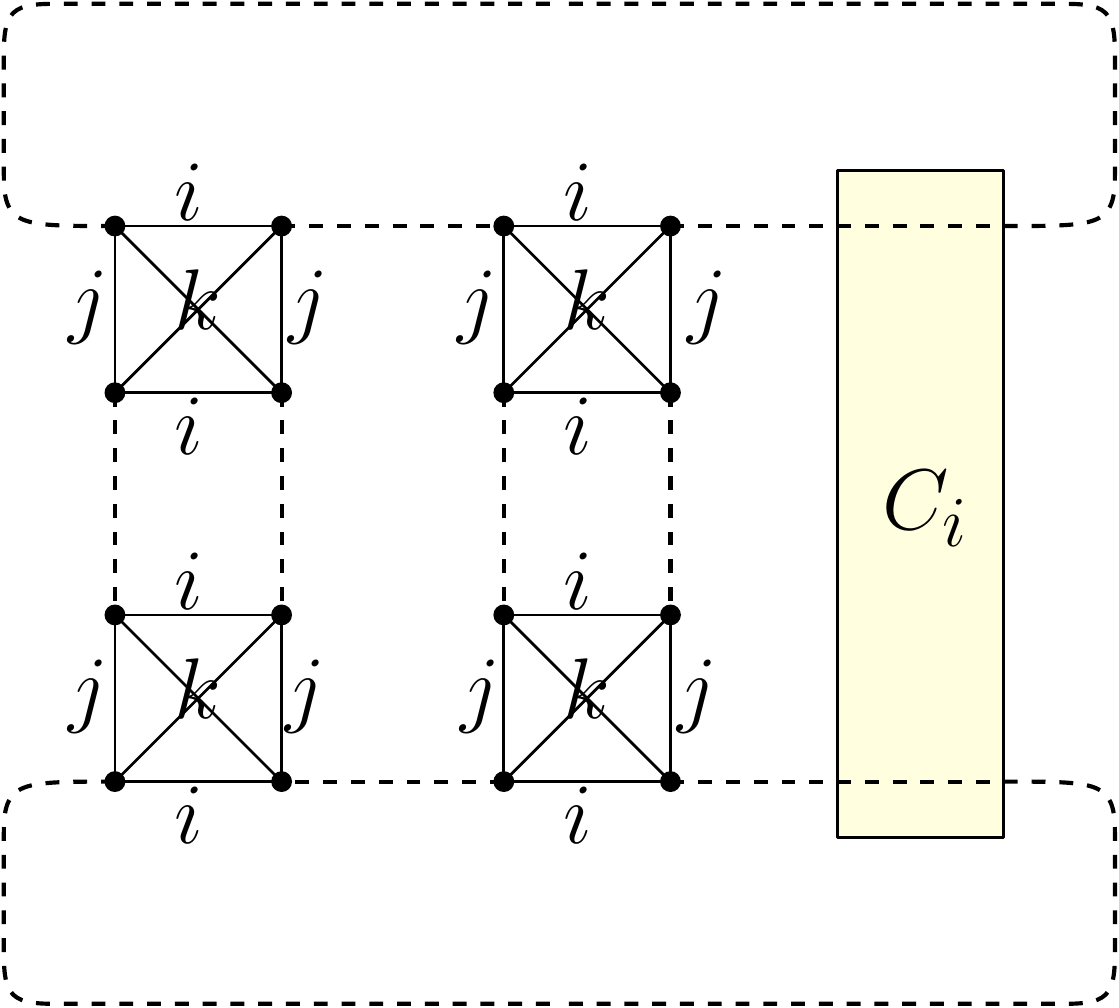} \end{array} \label{Family2Degree1}
\end{gather}
In this Theorem, dashed edges represent {\bf any melonic 2-point function}, including in \eqref{deg12PI}, and the boxes represent possibly empty chains of dipoles.
\end{theorem}

We have drawn explicitly two dipoles of color $i$ in \eqref{Family2Degree1} so that when the chain $C_i$ is empty the graph has 4 bubbles and is indeed of degree 1 (if it had only two bubbles it would be melonic).

\begin{proof}
We use an induction on the number of bubbles. 

\paragraph{Case $G$ 2PR.} Then we perform the flip as in \eqref{2EdgeCutDisconnected} which turns $G$ into the pair $\{G_L, G_R\}$ such that $\omega(G) = \omega(G_L)+\omega(G_R)$.
\begin{itemize}
\item If $G_L$ (or $G_R$) has vanishing degree, the other graph has degree 1 and fewer bubbles so we can apply the induction hypothesis. {The graph} $G_L$ is melonic and we find that $G$ is according to the theorem upon re-inserting $G_L$.
\item If $\omega(G_L)=\omega(G_R)=1/2$, then they are described by \eqref{DoubleTadpole}: double tadpoles with arbitrary melonic insertions on the edges of color 0. {The graph} $G$ is then a composition of $G_L$ and $G_R$. From Lemma \ref{thm:EdgeCutting} it is easy to see that this gives rise to the family \eqref{DoubleTadpoleChain}.
\end{itemize}
All the other cases correspond to $G$ being 2PI.

\paragraph{Case $G$ 2PI and dipole-free.} Then $G$ is exactly the graph \eqref{deg12PI}.

\paragraph{Case $G$ 2PI with a dipole.} Then the dipole can be eliminated as in \eqref{DipoleRemoval} leading to $G'$ connected, with two bubbles less than $G$. As seen in Proposition \ref{thm:Dipoles}, dipole removals decrease or preserve the degree. Since $G$ has degree 1, only the third and fourth cases of the proof of Theorem \ref{thm:Dipoles} may appear. They are the cases where the Lemma \ref{thm:NumberColors} and {Lemma} \ref{thm:Commuting} apply.
\begin{itemize}
\item If the dipole insertion decreases the degree by 1, then $G'$ is melonic. It is thus obtained by melonic insertions on the 2-bubble graph \eqref{TwoBubbleGraph}. From Lemma \ref{thm:Commuting}, we can study the dipole insertions directly on \eqref{TwoBubbleGraph}. Then, for any pair of edges of color 0 in \eqref{TwoBubbleGraph}, it is easy to check that a dipole insertion is in fact a melonic insertion, which in particular preserves the degree. 
\item Therefore the dipole removal from $G$ must preserves the degree: $\omega(G')=1$. From the induction hypothesis, $G'$ has the form given in the theorem. The last step of the proof is thus to verify that performing a dipole insertion which preserves the degree on any of those graphs does not create a graph which is not already in this family. 
%Thanks to 
{Using} 
Lemma \ref{thm:Commuting}, we can consider the three types of graphs of the theorem without any melonic 2-point graph.

\begin{itemize}
\item $G'$ has the form \eqref{DoubleTadpoleChain}. First consider the cases with only two bubbles. There are two of them, starting with
\begin{equation} \label{2BubbleTadpole1}
G' = \begin{array}{c} \includegraphics[scale=.5]{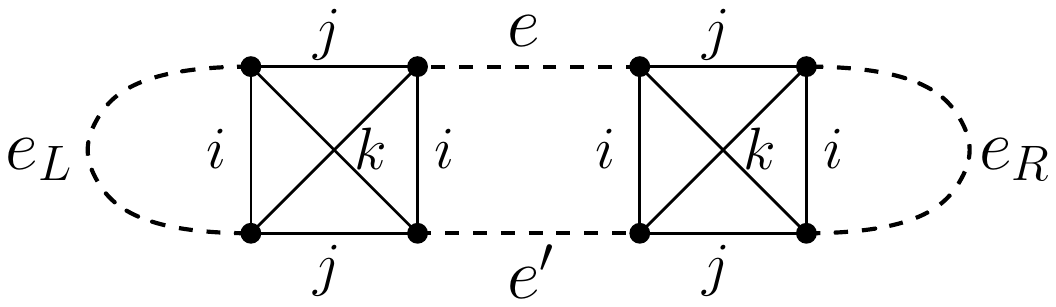} \end{array}
\end{equation}
where the same faces of color $j$ and $k$ go along $e_L$ and $e_R$. Then a dipole insertion which preserves the degree can be performed on these two edges and leads to the graph \eqref{Family2Degree1} with four bubbles, i.e.
\begin{equation} \label{Family2Degree1FourBubbles}
\begin{array}{c} \includegraphics[scale=.4]{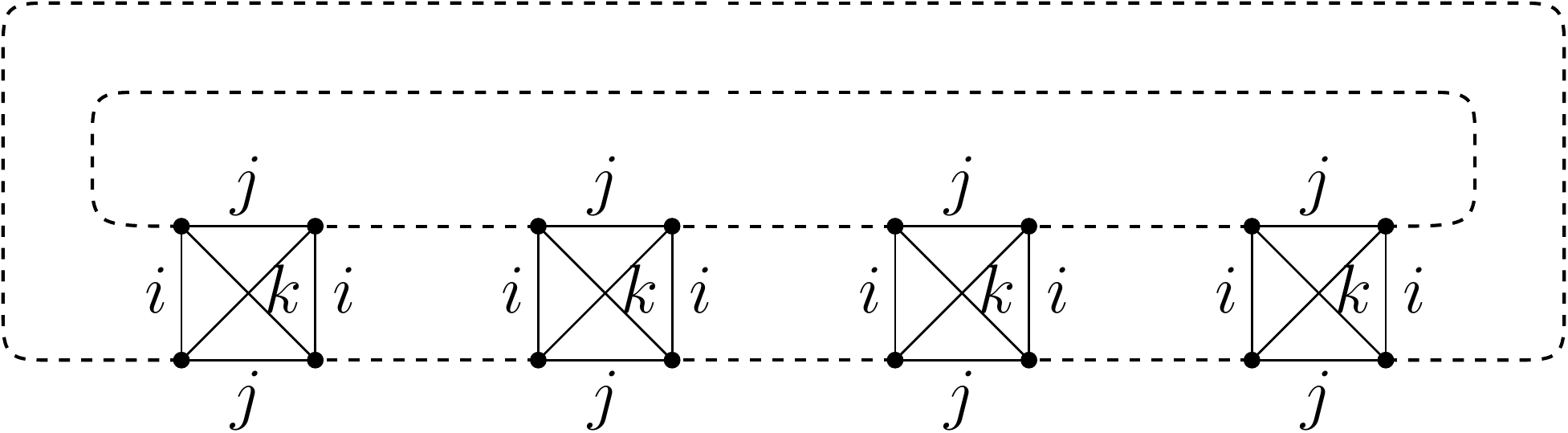} \end{array}
\end{equation}
 The other case is 
\begin{equation} \label{2BubbleTadpole2}
G' = \begin{array}{c} \includegraphics[scale=.5]{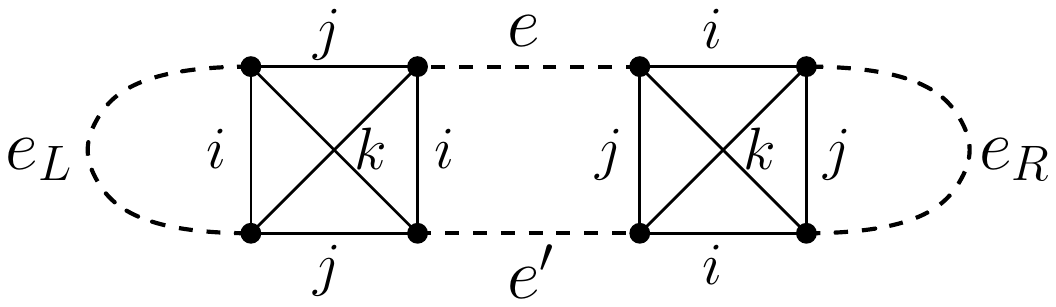} \end{array}
\end{equation}
where no dipole insertion preserving the degree can be done on $e_L, e_R$ this time. However, both in \eqref{2BubbleTadpole1} and \eqref{2BubbleTadpole2}, 
{the edges} 
$e$ and $e'$ have in common their faces of each color and a dipole of any color can be inserted without changing the degree. This 
{leads to} 
 a graph 
{such as} 
  \eqref{DoubleTadpoleChain}{,} with four bubbles.

Consider now a graph \eqref{DoubleTadpoleChain} with a non-empty chain. The dipole in position $l$ has color $j_l$,
\begin{equation}
\begin{array}{c} \includegraphics[scale=.5]{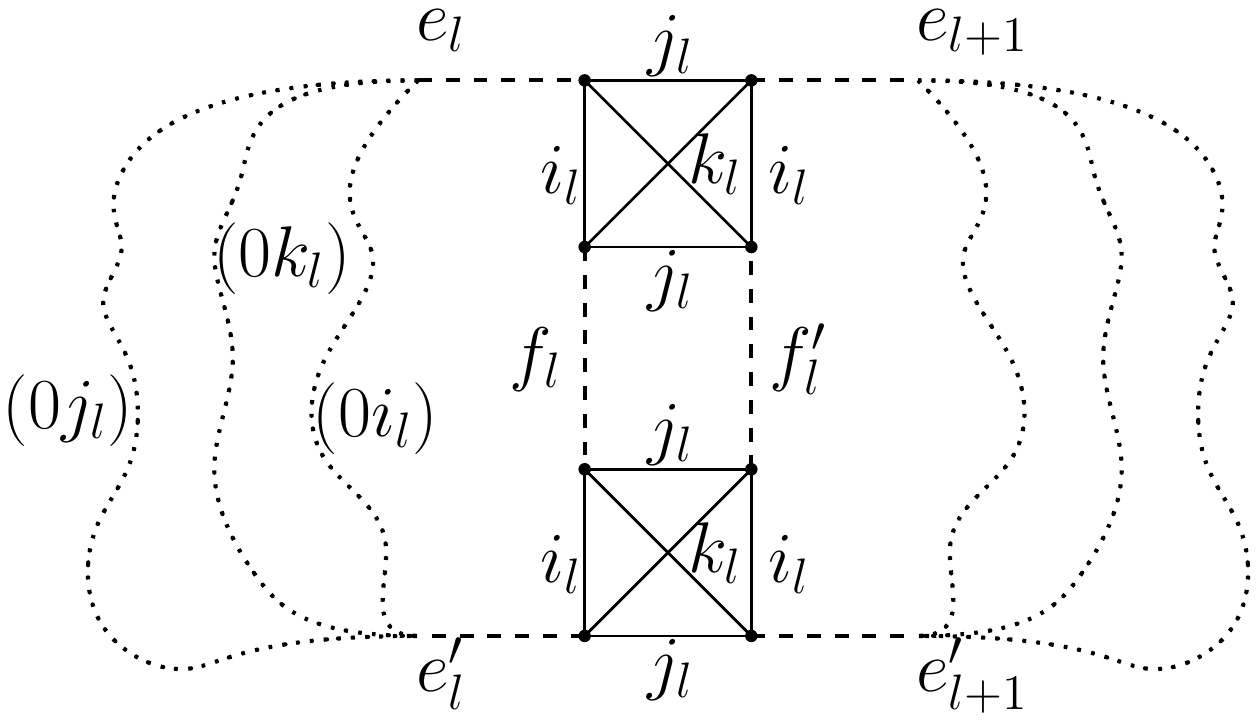} \end{array}
\end{equation}
The dotted lines represent the bicolored paths which close the faces of each color going along $e_l, e_{l+1}, e'_l, e'_{l+1}$. There can be dipole insertions which preserve the degree on the pair $\{e_l, e'_l\}$ and on the pair $\{e_{l+1}, e'_{l+1}\}$, both resulting in lengthening the chain of dipoles. These dipole insertions can be of any color. Notice that the dotted paths on the left are disjoint from those on the right. This implies that no other dipole insertions preserving the degree can be done.

\item  {The graph} $G'$ has the form \eqref{Family2Degree1}. First, it can be the graph with four bubbles pictured in \eqref{Family2Degree1FourBubbles}. Only one dipole preserving the degree can be inserted, of color $i$
\begin{equation}
\begin{array}{c} \includegraphics[scale=.4]{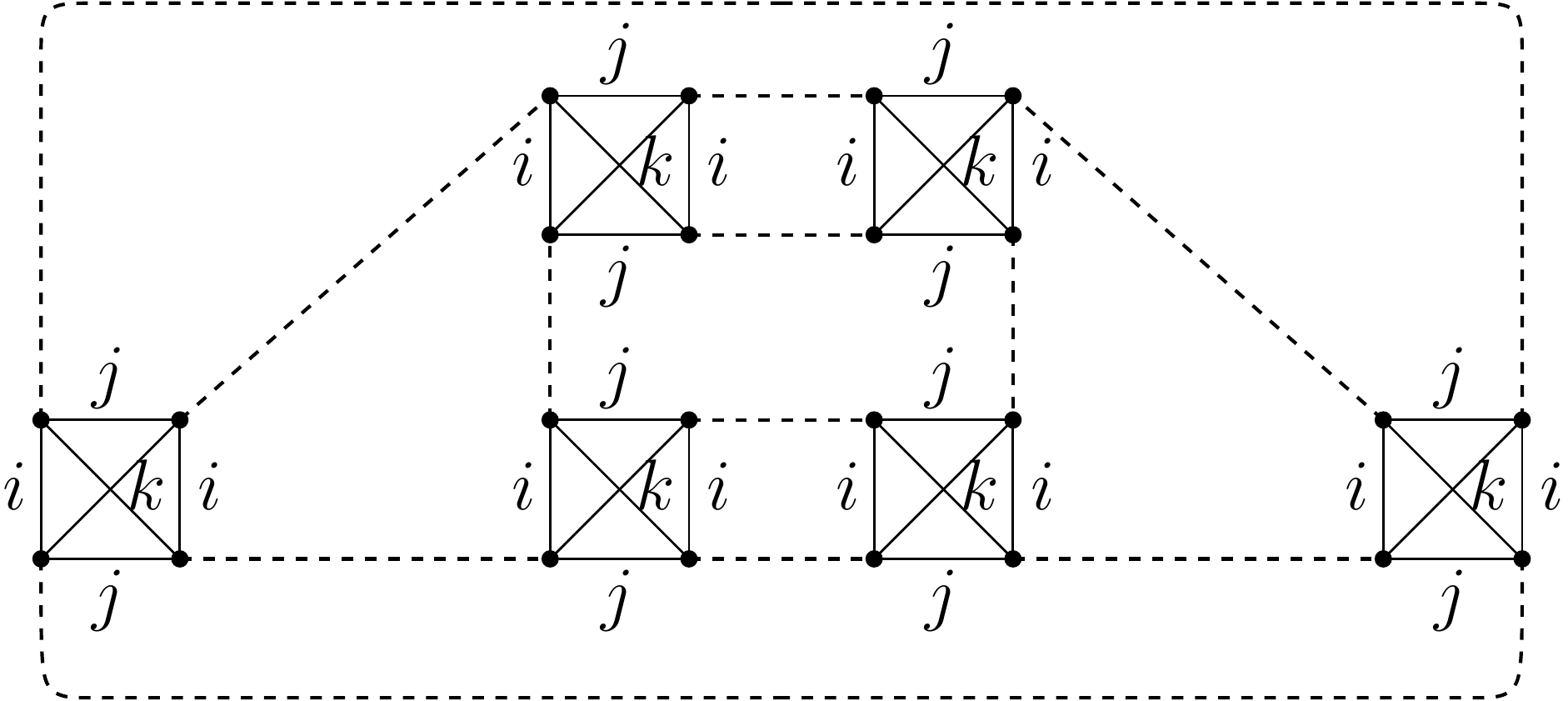} \end{array}
\end{equation}
One can then add dipoles of color $i$ to create a chain of dipoles. This corresponds to \eqref{Family2Degree1}.

No other dipole preserving the degree can be inserted. Indeed, a dipole of the chain is 
%as follows
 {of the type}
\begin{equation}
\begin{array}{c} \includegraphics[scale=.5]{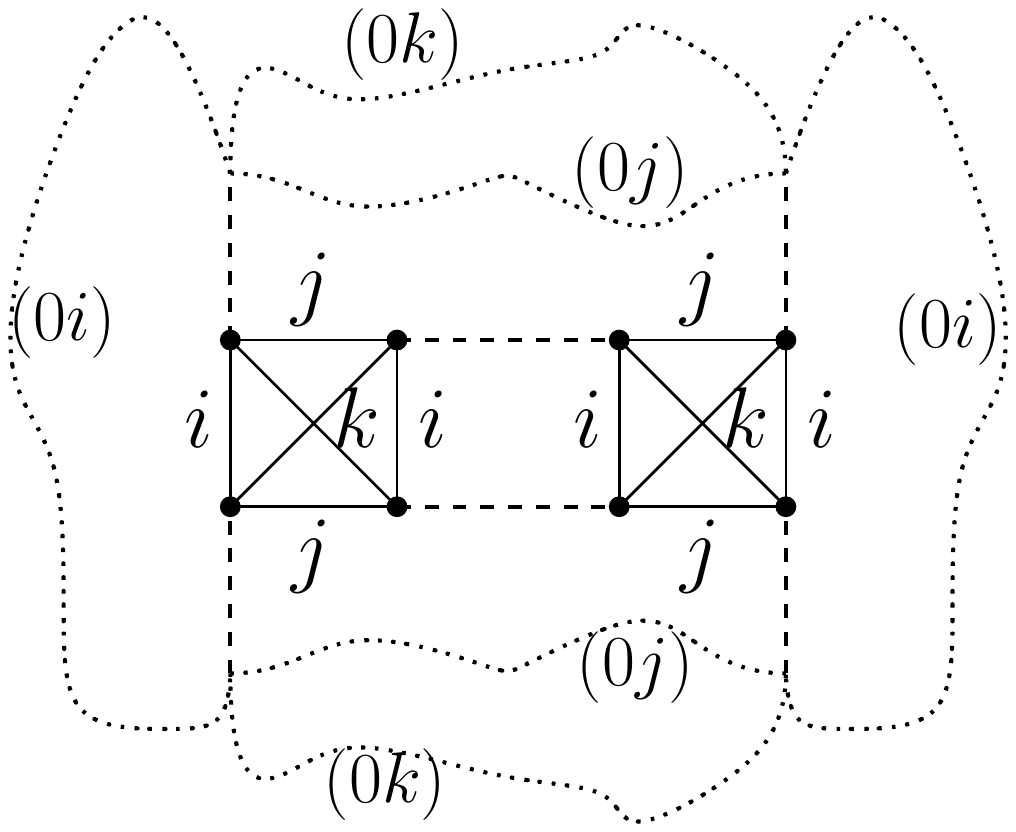} \end{array}
\end{equation}
where the dotted lines represent the bicolored paths closing the faces. Those paths are disjoint and therefore only dipoles of color $i$ extending the chain are allowed.

\item {The graph} $G'$ can be the graph \eqref{deg12PI}. It can be checked directly that no two edges have more than one face in common, meaning that no dipole insertion preserving the degree exists.
\end{itemize}
\end{itemize}
This ends the induction {and thus concludes the proof}.
\end{proof}

%%%%%%%%%%%%%%%%%%%%%%
\section{Degree $3/2$ graphs of the \ON invariant SYK-like tensor model} \label{sec:Degree3/2}
%%%%%%%%%%%%%%%%%%%%%%

Our strategy can be applied to graphs of higher degrees, provided one can identify the 2PI, dipole-free graphs (step \ref{enum:2PIDipoleFree} of the strategy in Section \ref{sec:Strategy}). We do so in the case of 
{degree} 
$\omega=3/2$ below. The other steps of the strategy are briefly discussed at the end 
{of the section}
and are left to be completed 
%as exercises.
{by the interested reader}.

\begin{theorem} \label{thm:2PIDipoleFreeDegree3/2}
There is a{n} unique \ON-invariant, dipole-free, 2PI graph of degree $3/2$ which is
\begin{equation}
\begin{array}{c} \includegraphics[scale=.5]{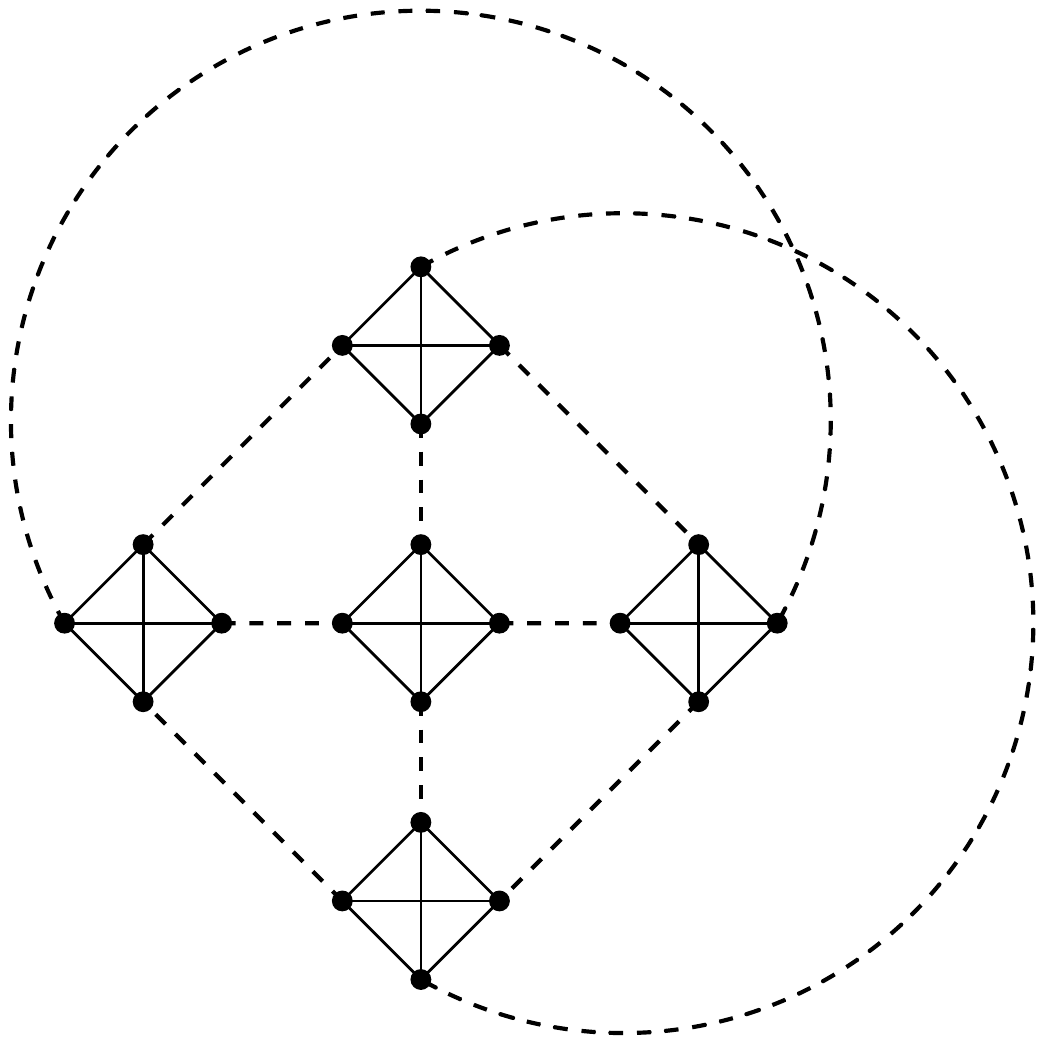} \end{array}
\end{equation}
\end{theorem}

Its three jackets have genus $1/2$ (topological projective planes). For each color, there are exactly 2 faces of length three and one of length four.
%, which will be key elements of the proof.

\begin{proof}
Let $G$ be a graph of degree $3/2$. {\it A priori}, its jackets can have genera
\begin{equation}
(g_1, g_2, g_3) = (3/2,0,0) \qquad \text{or} \qquad (g_1, g_2, g_3) = (1,1/2,0) \qquad \text{or} \qquad (g_1, g_2, g_3) = (1/2,1/2,1/2)
\end{equation}
up to color permutations. The first two cases have at least one planar jacket, while the third case has not. We first show that the jackets cannot be planar.

\paragraph{Assume one jacket is planar,} \label{sec:2PI_deg32_planar} say $J_3$. From Lemma \ref{thm:EvenLength}, the faces of color 3 have even lengths, $F_{3,2l+1}=0$. This in turn simplifies \eqref{cond_face} to
\begin{equation}
\sum_{l\geq 2} (l-2) F_{3,2l} = 1,
\end{equation}
further implying that there is one face of length 6, all others being of length 4.

It can be checked that the face of length 6 has to have six distinct bubbles. For instance, we investigate the cases where it has only five bubbles. It could be as follows
\begin{equation}
\begin{array}{c} \includegraphics[scale=.5]{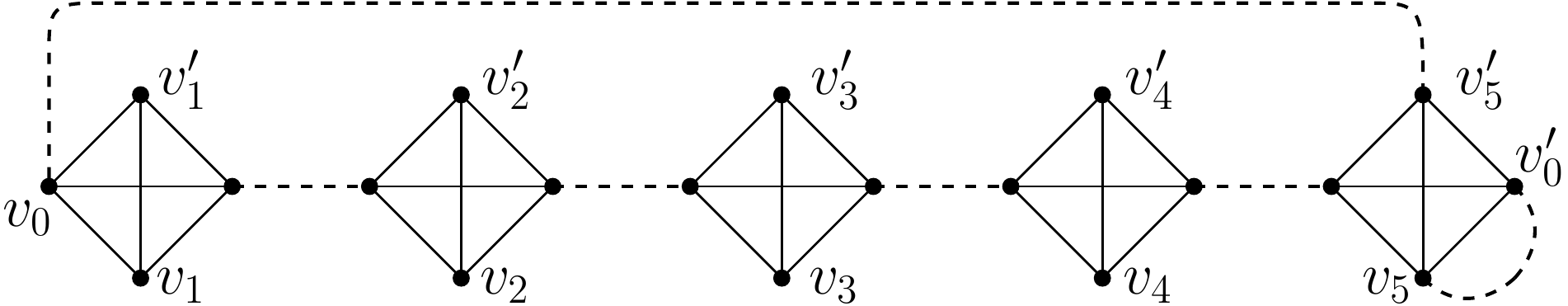} \end{array}
\end{equation}
where $v_0$ is connected $v'_5$, but then it forces a face of length 1 on the color 1 or 2 and $G$ would be 2PR.

If $v_0$ is connected to $v'_4$, then the planarity of $J_3$ forces $v_5$ to be connected to a 1-point function, and $v'_1, v'_2, v'_3$ to a 3-point function,
\begin{equation}
\begin{array}{c} \includegraphics[scale=.5]{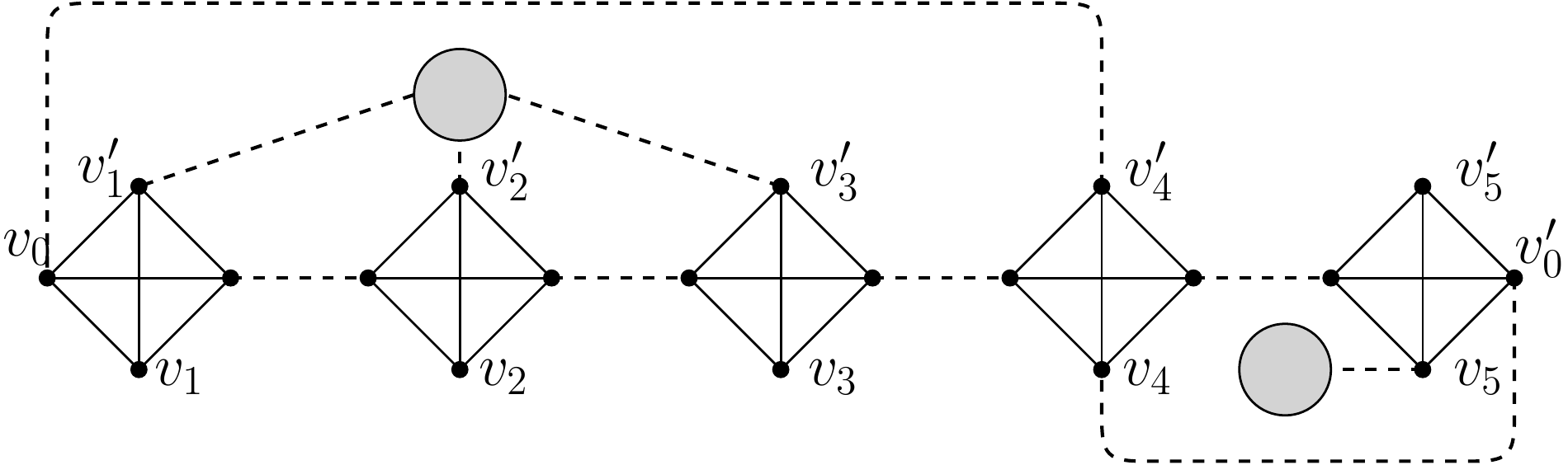} \end{array}
\end{equation}
but they do no exist since the number of vertices is always even (no odd-point functions).

If $v_0$ is connected to $v'_3$, then planarity of $J_3$ forces the edges of color 0 connected to $v_4$ and $v_5$ to form a 2-cut, making $G$ 2PR, or to connect $v_4$ to $v_5$ by a single edge of color 0 which then forms a dipole and this is forbidden as well,
\begin{equation}
\begin{array}{c} \includegraphics[scale=.5]{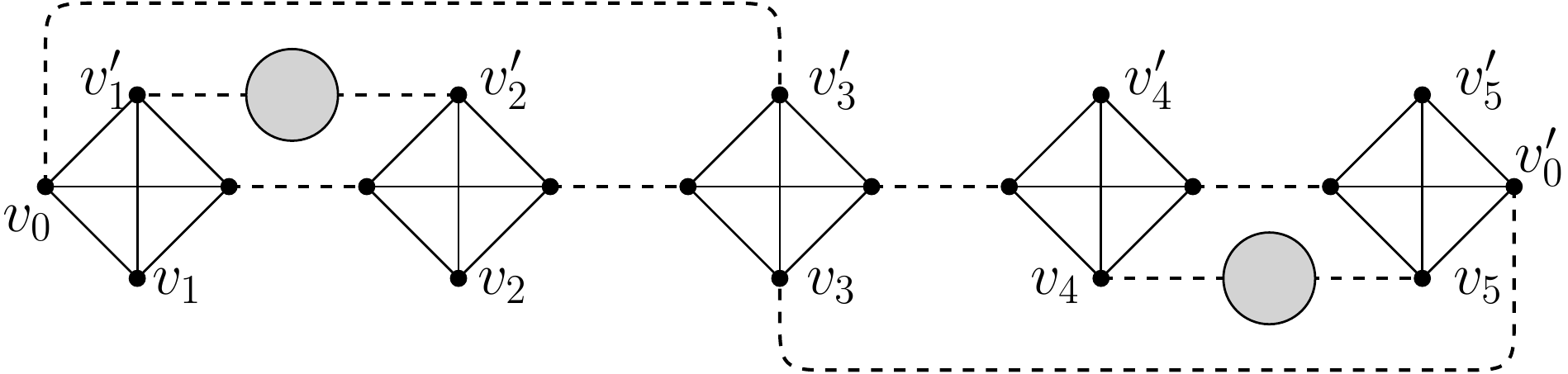} \end{array}
\end{equation}
(and similarly between $v'_1$ and $v'_2$).

Let us thus denote $b_1, \dotsc, b_6$ the six distinct bubbles forming a face of color 3 and length 6. The other edges of color 3 of those bubbles must belong to faces of length 4. The same arguments as those used just above to prove that the bubbles are distinct can be used to prove that the vertex $v_1$ of $b_1$ cannot be connected to a vertex of $b_2, \dotsc, b_6$. It therefore connects to another bubble $b_7$, and by symmetry $v'_1$ too, to $b_8$. A face of length 4 is then obtained by ``crossing'' one of the bubbles $b_2, \dotsc, b_6$. Say it crosses $b_3$,
\begin{equation}
\begin{array}{c} \includegraphics[scale=.5]{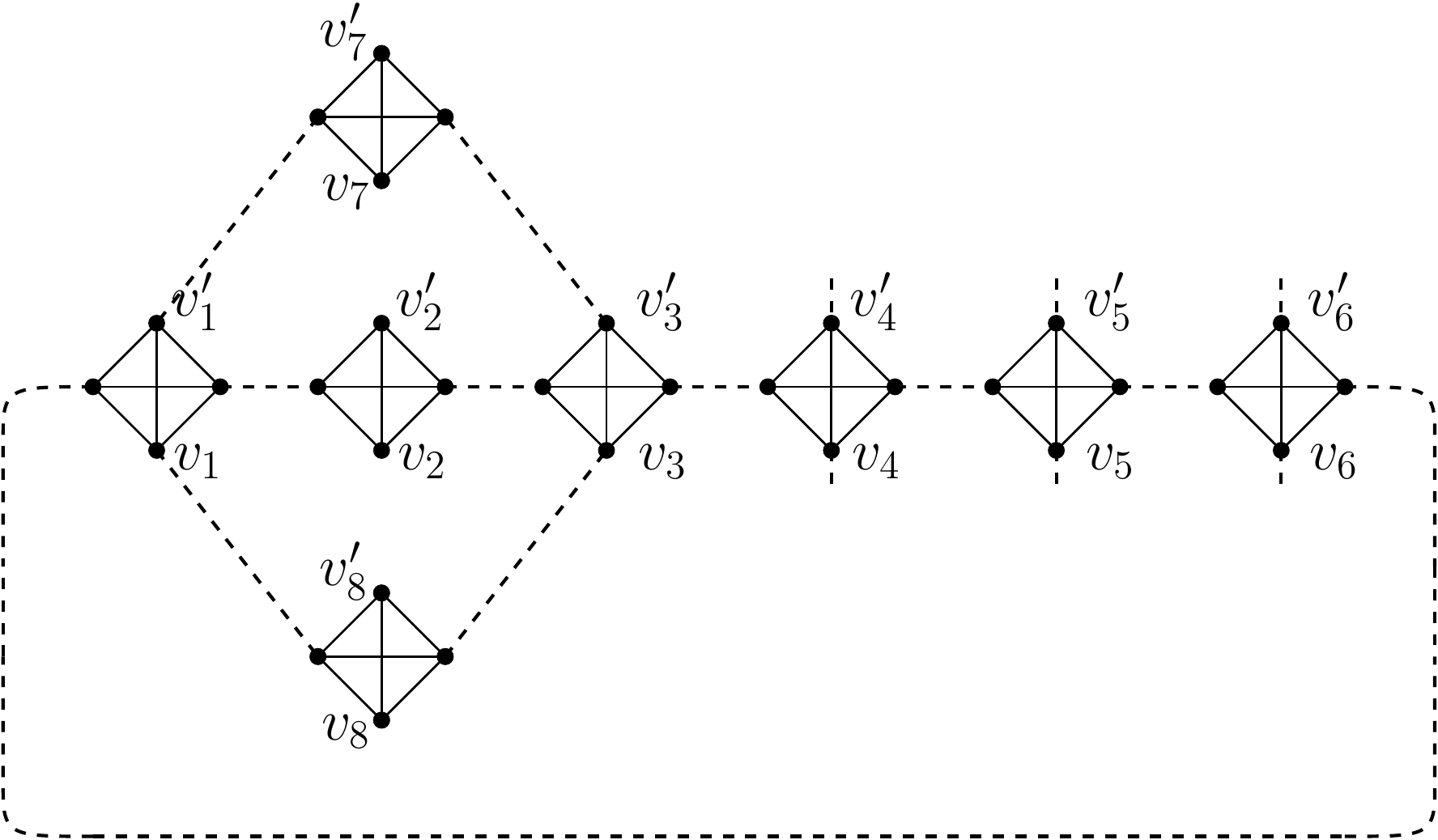} \end{array}
\end{equation}
then planarity of $J_3$ forces a 2-point function between $v'_2$ and $v_7$, and $G$ being 2PI forces this 2-point function to be trivial, i.e. a single edge of color 0. Same thing between $v_2$ and $v'_8$. The face of color 3 going along those edges must be of length 4 too. To do that, it is necessary to cross one of the bubbles $b_4, b_5$ or $b_6$. In the case it crosses $b_4$
\begin{equation}
\begin{array}{c} \includegraphics[scale=.5]{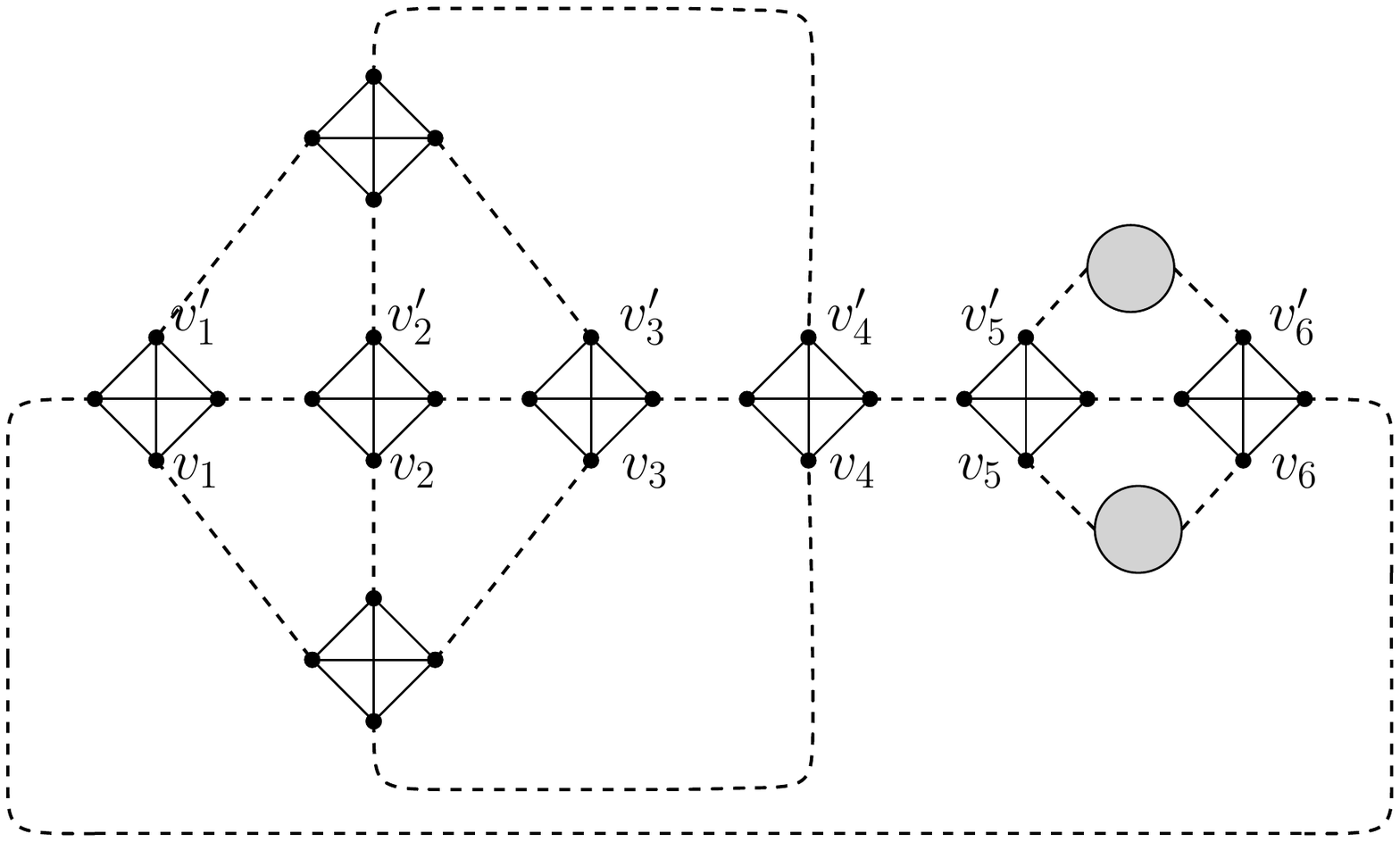} \end{array}
\end{equation}
Planarity of $J_3$ again forces a 2-point function between $v_5$ and $v_6$, making $G$ 2PR, or directly an edge of color 0 creating a dipole. If the face of color 3 going through $v_2$ and $v'_2$ crosses $b_6$, this is similar. If it crosses $b_5$ instead, planarity of $J_3$ would force 1-point function connected to $v_4$ and $v'_4$, which is impossible. 

All other cases are treated similarly. The conclusion is that $G$ 2PI, dipole-free, of degree $3/2$, cannot have a planar jacket.

\paragraph{All jackets therefore have genus $1/2$.} We can thus embed any jacket into the projective plane without crossings. Adding the edges of the missing color to the jacket, an embedding of $G$ is obtained. As a convention, we will consider that $J_3$ is embedded without crossings, and the projective plane is represented as a disc with opposite points identified.

Next thing is to study the face lengths. From Lemma \ref{FaceLength}, one finds
\begin{equation}
F_{i,3} = 2 + \sum_{l\geq 5}(l-4) F_{i,l} \geq 2
\end{equation}
meaning that there are at least two faces of length 3 for each color. 

\paragraph{The faces of length 3.} Consider a face of length 3 and color 3 and try to represent it using the embedding. It might be like
\begin{equation}
\begin{array}{c} \includegraphics[scale=.5]{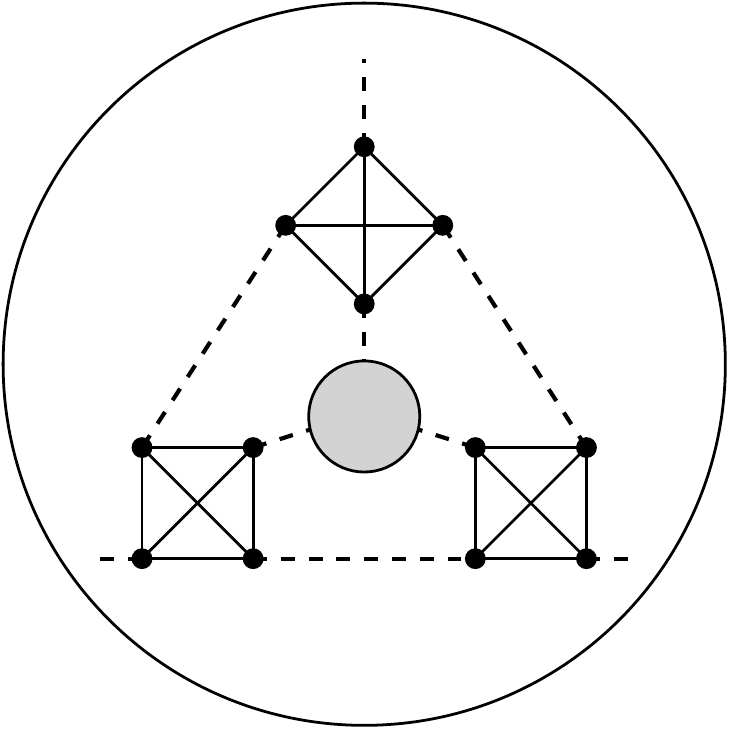} \end{array}
\end{equation}
where the interior regions must then have a 3-point function which is impossible. Therefore a face of length 3 and color 3 must follow the (unique up to homotopy) non-contractible cycle of the projective plane like
\begin{equation}
\begin{array}{c} \includegraphics[scale=.5]{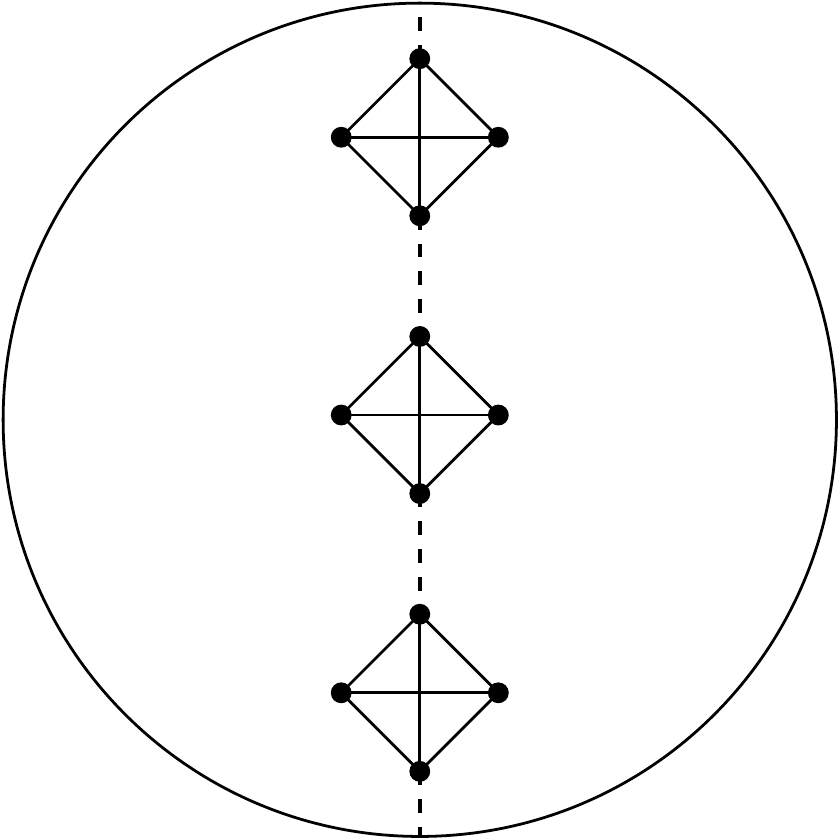} \end{array}
\end{equation}
We can then study the location of a second face of length 3 and color 3. It must obviously also wrap around the non-contractible cycle and must therefore ``cross'' the previous face. This gives
\begin{equation} \label{TwoFacesDegreeThree}
\begin{array}{c} \includegraphics[scale=.5]{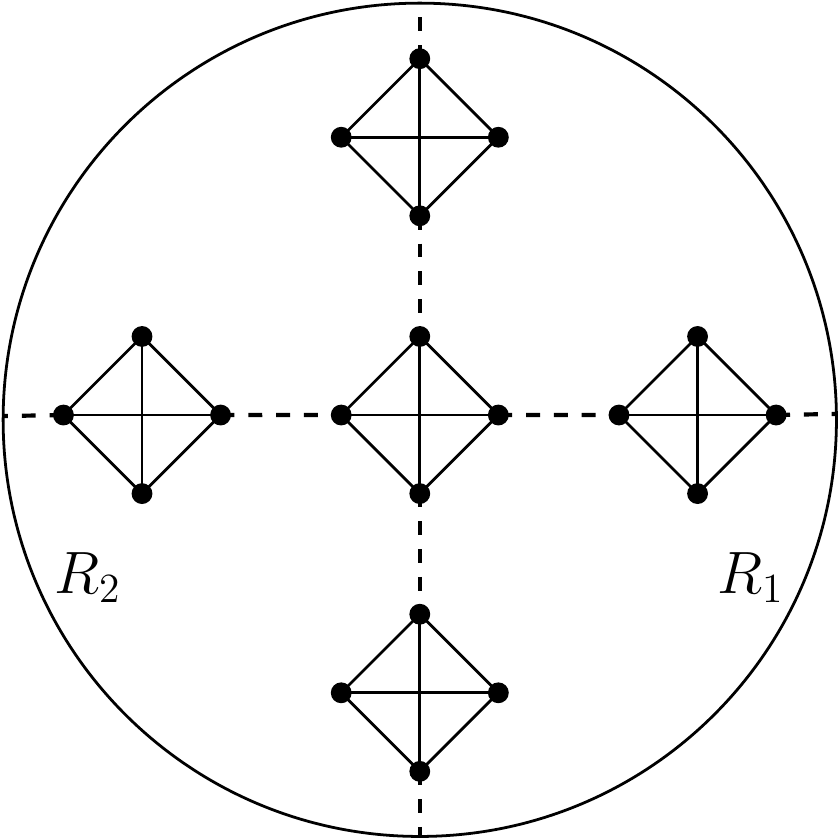} \end{array}
\end{equation}
In $J_3$, the exterior of the bubbles consists in two regions $R_1$ and $R_2$. A third face of color 3 and length 3 would have to go through both regions and thus be like
\begin{equation}
\begin{array}{c} \includegraphics[scale=.5]{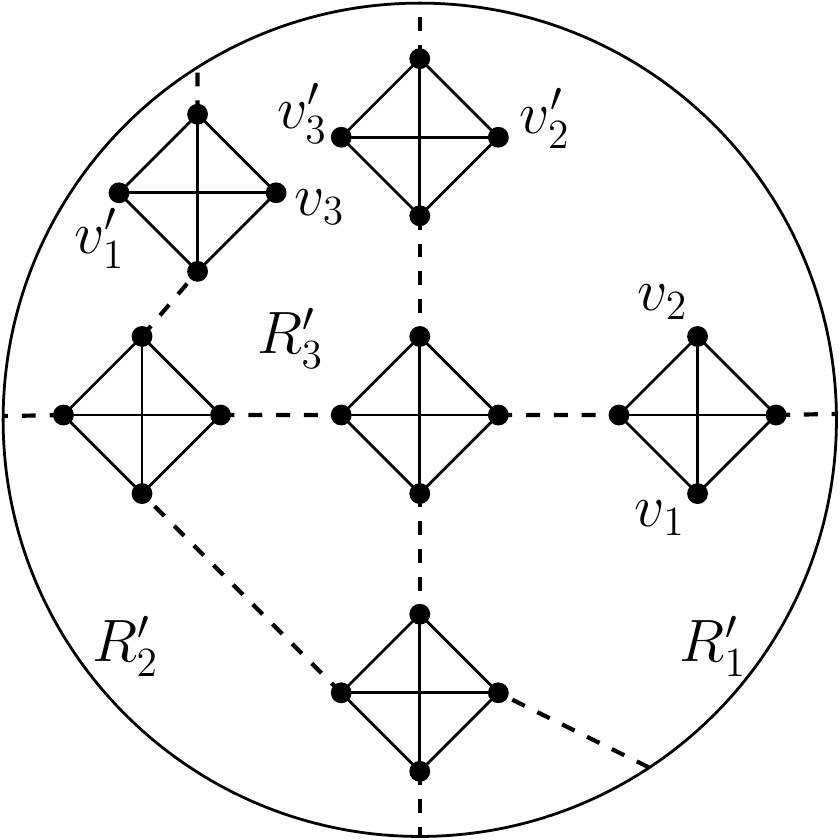} \end{array}
\end{equation}
The only unpaired vertices in the region $R'_1$ are $v_1$ and $v'_1$. Only a 2-point function can thus connect them and it in fact has to be just an edge of color 0 for $G$ to be 2PI. For the same reason, $v_2$ and $v'_2$ must be connected by an edge of color 0, as well as $v_3$ and $v'_3$. All vertices are then adjacent to edges of color 0. However the counting of faces reveals that this graph is of degree 1 (its jacket $J_2$ is planar). We conclude that there are only two faces of color 3 and length 3. The other faces of color 3 are thus all of length 4.

\paragraph{The faces of length 4.} Consider one, denoted $f$, going through one of the bubbles already drawn in \eqref{TwoFacesDegreeThree}. It has to go through the two regions $R_1$ and $R_2$ and must therefore go through at least two of the bubbles already drawn. 

Say it goes through $v_1$ and $v'_2$, with new notations depicted below, with an additional bubble between them. Then the vertex $v$ of the new bubble lies in a region with no other unpaired vertex, so it has to connect to a 1-point function which is zero,
\begin{equation}
\begin{array}{c} \includegraphics[scale=.5]{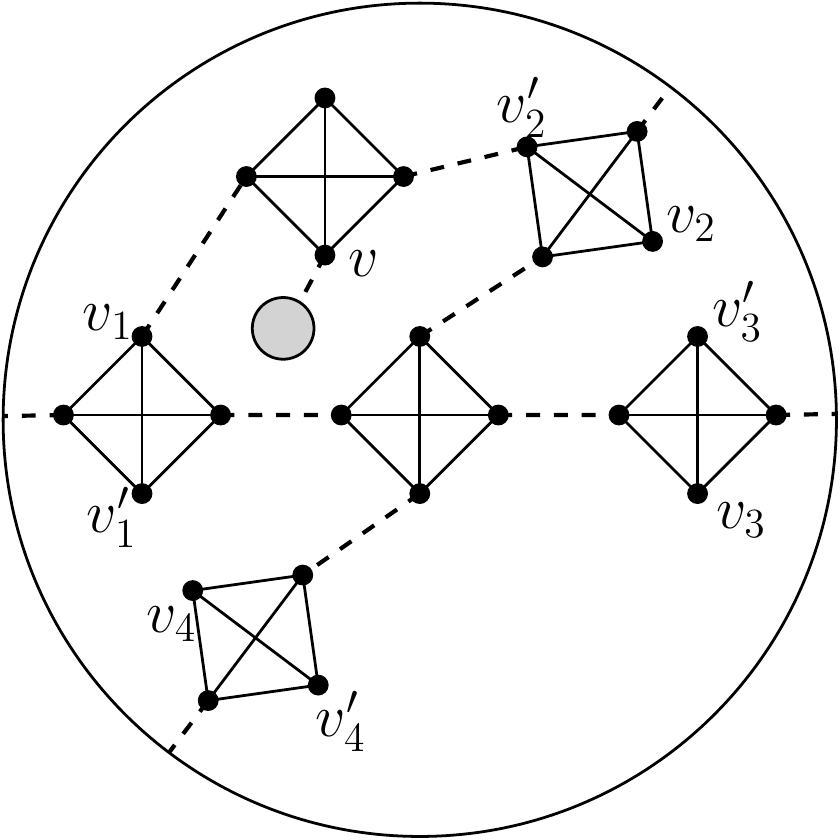} \end{array}
\end{equation}
Similarly, if there are two new bubbles between $v_1$ and $v'_2$, then their vertices must form a non-trivial 2-point function or a dipole, which are both forbidden. It goes similarly if one tries to add a bubble between $v_1$ and $v_3$ with a path of colors 0 and 3 between them, 
\begin{equation}
\begin{array}{c} \includegraphics[scale=.5]{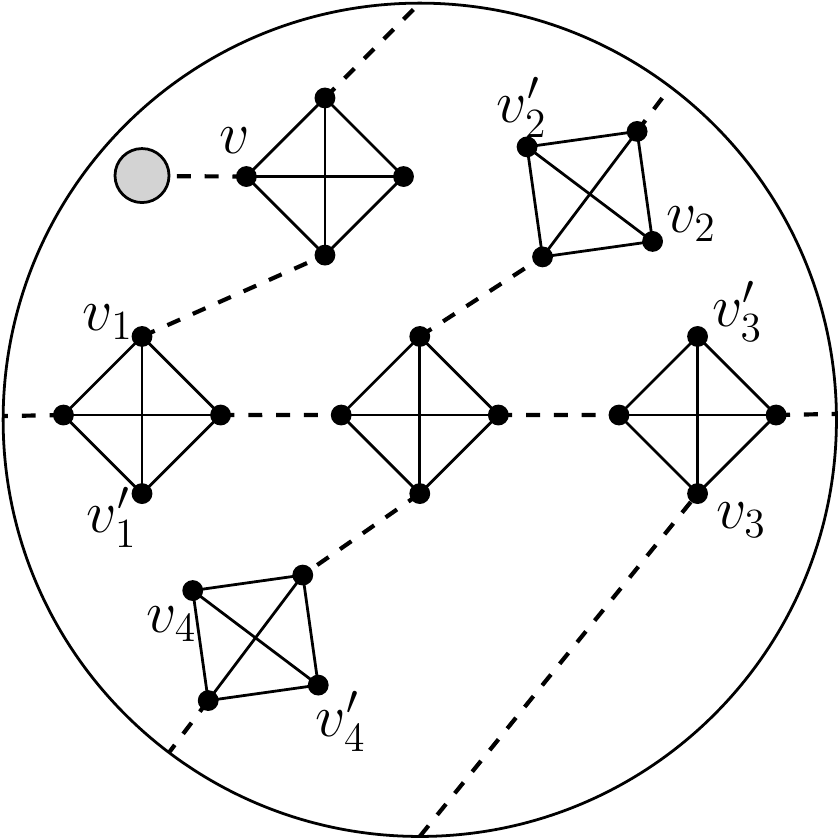} \end{array}
\end{equation}
with a 1-point function, and if two bubbles are inserted, they must create a non-trivial 2-point function or a dipole. 

Finally, let us try and add a bubble connecting $v_1$ to $v'_4$ via a path of colors 0 and 3. The embedding of $J_3$ in the projective plane enforces a trivial 2-point function between $v$ and $v'_2$ and between $v'$ and $v'_3$, and a 4-point function in the last region (it may be two 2-point functions)
\begin{equation}
\begin{array}{c} \includegraphics[scale=.5]{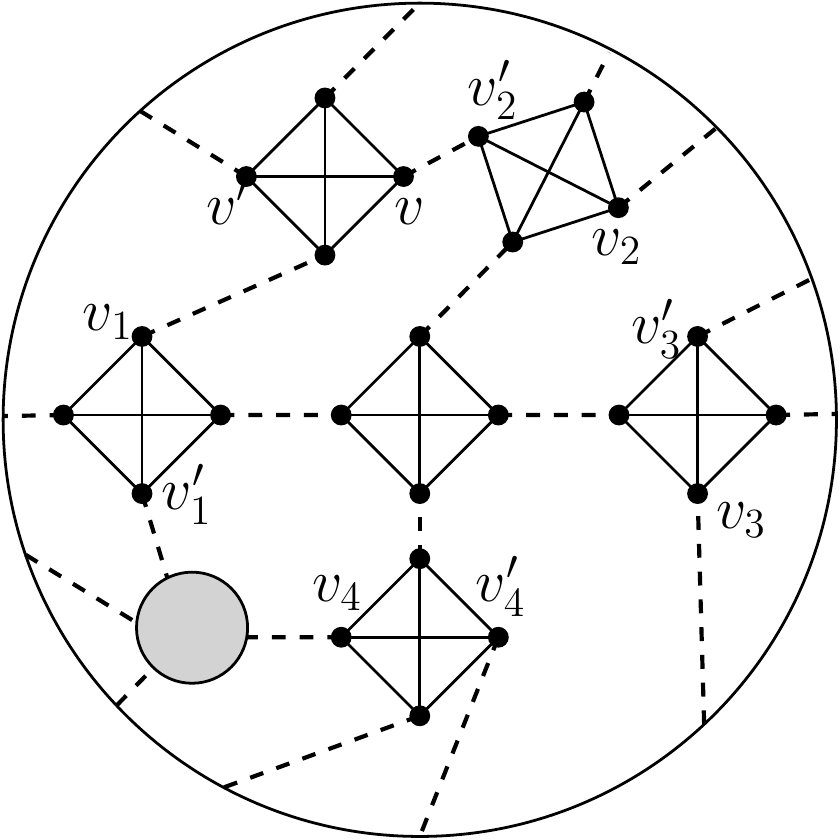} \end{array}
\end{equation}
However, there is no way to close the two open faces of color 3 so that they are of length 4. Indeed, this 4-point function has to contain a path of colors 0 and 3 between $v'_1$ and $v_4$ with exactly one bubble, which then enforces a 1-point function. If two bubbles connecting $v_1$ to $v'_4$ via a path with colors 0 and 3, then it is direct to see that this enforces 3-point functions in the regions $R$ and $R'$ below
\begin{equation}
\begin{array}{c} \includegraphics[scale=.5]{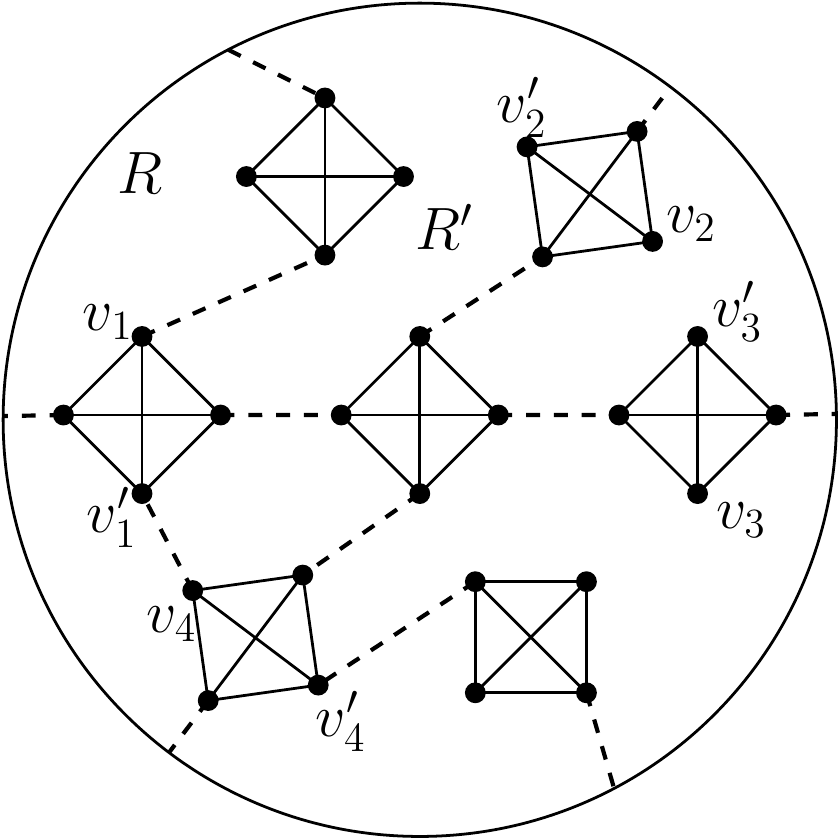} \end{array}
\end{equation}
and this is thus impossible.

This leaves no other option than to connect the unpaired vertices of \eqref{TwoFacesDegreeThree} without adding other bubbles and while maintaining the embedding of $J_3$ in the projective plane. If there is an edge of color 0 between $v_1$ and $v_3$, this creates a dipole. If $v_1$ is connected to $v'_4$ instead, then one might connect $v_4$ to $v_2$ but this creates a dipole. Instead, $v_4$ can be connected to $v'_3$ but then $v_3$ needs a 1-point function to avoid crossings, which is impossible,
\begin{equation}
\begin{array}{c} \includegraphics[scale=.5]{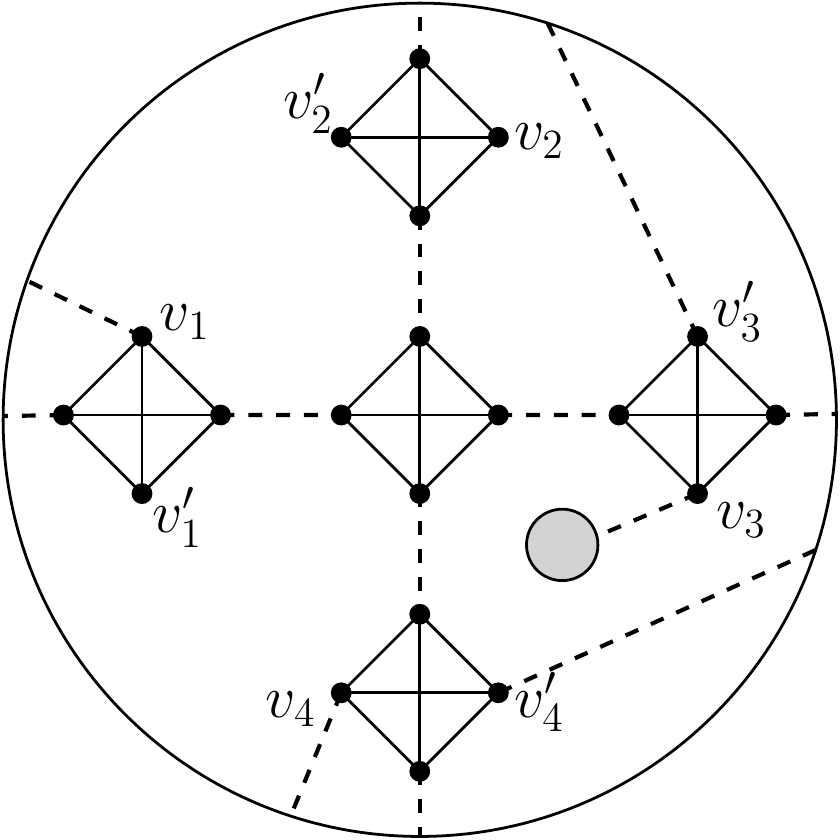} \end{array}
\end{equation}
Therefore $v_1$ can only {be} connected to $v'_2$. Then $v_2$ cannot be connected to $v_4$ without forming a dipole, and cannot be connected to $v'_1$ because the face of color 3 would have length 2 only. Thus, $v_2$ can only be connected to $v'_3$, and then $v_3$ to $v'_4$ and finally $v_4$ to $v'_1$
\begin{equation}
\begin{array}{c} \includegraphics[scale=.5]{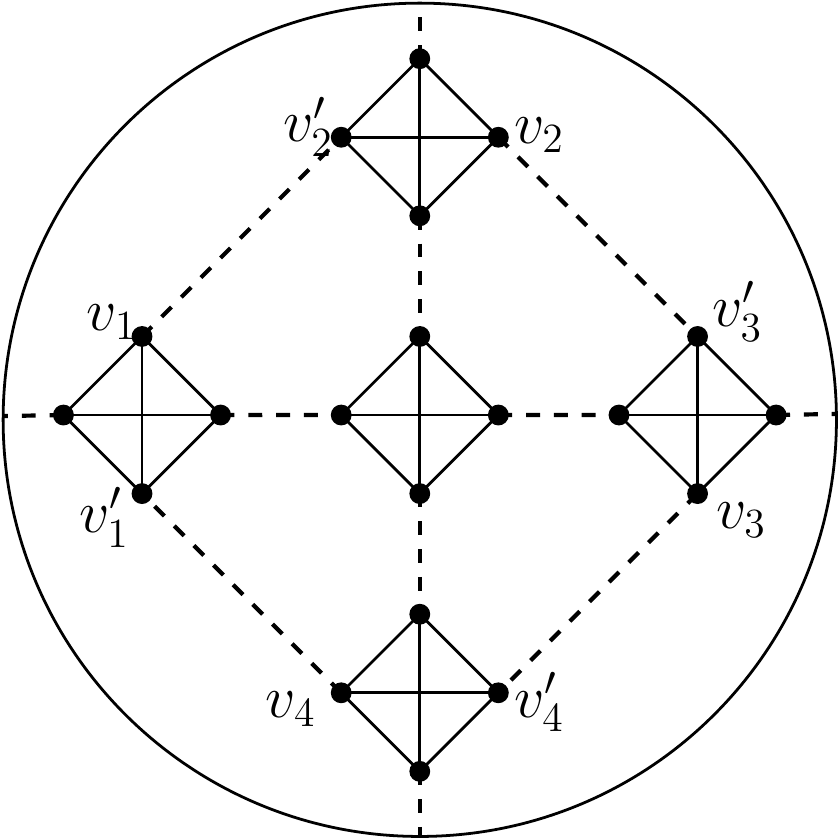} \end{array}
\end{equation}
which is the graph of the Theorem. {This concludes the proof.}
\end{proof}

\bigskip

Let us 
%a few words 
{end this section by adding a few comments}
on the other steps of the strategy to find all graphs of degree 3/2.
\begin{itemize}
\item As for 2PR graphs without melonic insertions (Step \ref{enum:2PR}), it is enough to consider all compositions of $G_L$ and $G_R$ with $\omega(G_L) = 1$ and $\omega(G_R)=1/2$. Those graphs have been detailed throughout this article. We leave the compositions to the {interested} reader.
\item As for the dipole insertions, as in the case of degree 1, only those which preserve the degree and those which increase the degree by one have to be considered.
\end{itemize}

\section{Concluding remarks} \label{sec:Conclusion}
%%%%%%%%%%

Let us now recall that in \cite{Fusy1}, the dominant graphs at any order in the large $N$ expansion of the MO model have been studied. 
%These dominant configurations have been %, limit which lies outside the aims of this article.
We have compared in this paper the diagrammatics of the two SYK-like tensor models, the MO model and the CTKT, or \ON-invariant one. Our main results are{:}
\begin{itemize}
\item an identification of a subset of \ON-invariant graphs which correspond to MO graphs: \ON-invariant graphs with an orientable jacket are MO with the same degree.
\item a recipe to identify graphs of a fixed degree $\omega$, outlined in Section \ref{sec:Strategy} and applied to the graphs of degree 1 and and 3/2.
%(with steps left to readers in the latter case).
\end{itemize}

%We hope that 
{A first perspective for future work is to apply} 
our strategy to $n$-point functions, which are central in SYK-like models (see \cite{last})
{and to thus}
 find the relevant graphs at a given order in the $1/N$ expansion.

Given the strong relation we proved between the MO and \ON-invariant models, it is natural to compare our results with the combinatorial analysis of the MO model performed in \cite{Fusy1}. There, graphs of arbitrary degree were considered, but only the ``most singular'' ones which were shown to arise from rooted binary trees (see Proposition $22$ of \cite{Fusy1}). It then allowed for the implementation of the double scaling mechanism for the MO tensor model, see \cite{donald}. In the present paper (independently of the fact that we study here the more general \ON-invariant model), we instead focus on a strategy which identifies all graphs at fixed $\omega$, and not some most singular ones. The cost is that it seems difficult to apply it to an arbitrary order $\omega$, but this is a{nother} direction for future work.

Let us also note that that type of analysis of the general term in the large $N$ expansion has already been done in \cite{Fusy2} (see also \cite{BLT} for a{n} LO and NLO analysis) for the colored SYK model (which is a particular case of the SYK generalization introduced in \cite{Gross}). Therefore, another interesting perspective for future work is the comparison of such an analysis of the general term of \ON-invariant SYK-like tensor model with the results of the analysis of \cite{Fusy2} for the colored SYK model.

%%%%%%%%%%%%%%%%%%%
\begin{acknowledgments}
Valentin Bonzom and Adrian Tanasa are partially supported by the CNRS Infiniti "ModTens" grant. Adrian Tanasa is partially supported by the PN 09 37 01 02 grant. Victor Nador is fully supported by the CNRS Infiniti "ModTens" grant. Valentin Bonzom is supported by the ANR MetAConc project ANR-15-CE40-0014.
\end{acknowledgments}
%\section*{Acknowledgments}

\bibliographystyle{spmpsci}      % mathematics and physical sciences
\bibliography{bib_CTKT}     % name your BibTeX data base
\end{document}